\documentclass[1p]{elsarticle}
\usepackage{hyperref}
\journal{Applied and Computational Harmonic Analysis}
\usepackage{pifont}
\usepackage[cmex10]{amsmath}
\usepackage{graphicx}
\usepackage{graphics,picins}
\usepackage{tikz}
\usepackage{wrapfig}
\usepackage{xfrac}
\usepackage{amsfonts,amssymb,amscd,amsthm}
\usepackage{mathtools}
\usepackage{algorithm}
\usepackage{algpseudocode}
\usepackage{url}
\usepackage{bookmark}
\usepackage{stmaryrd}

\newdefinition{definition}{Definition}
\newdefinition{remark}{Remark}
\newtheorem{proposition}{Proposition}
\newdefinition{example}{Example}
\usepackage{verbatim,color,colortbl}

\usepackage{xcolor}

\newcommand{\etal}{\emph{et al.} }
\newcommand{\eg}{\emph{e.g.}}
\newcommand{\ie}{\emph{i.e.}}
\DeclareMathOperator{\Tr}{Tr}
\DeclareMathOperator{\MSE}{MSE}
\DeclareMathOperator{\diag}{diag}

\newcommand\norm[1]{\left\lVert#1\right\rVert}
\newcommand\abs[1]{\left\lvert#1\right\rvert}

\def\bR{{\mathbb {R}}}
\def\bC{{\mathbb {C}}}

\def\bN{{\mathbb {N}}}
\def\bZ{{\mathbb {Z}}}

\def\bn{{\mathbf{n}}}

\def\bz{{\mathbf{z}}}

\def\0{{\mathbf{0}}}
\def\1{{\mathbf{1}}}

\def\bb{{\mathbf{b}}}
\def\bc{{\mathbf{c}}}
\def\be{{\mathbf{e}}}
\def\bx{{\mathbf{x}}}
\def\by{{\mathbf{y}}}
\def\BI{{\mathbf{I}}}

\def\BA{{\mathbf{A}}}
\def\BC{{\mathbf{C}}}
\def\BD{{\mathbf{D}}}

\def\BB{{\mathbf{B}}}
\def\EE{{\mathrm{E}}}
\def\BE{{\mathbf{E}}}

\def\BU{{\mathbf{U}}}

\def\BY{{\mathbf{Y}}}

\def\BS{{\mathbf{S}}}
\def\cP{{\mathcal{P}}}
\def\BP{{\mathbf{P}}}
\def\BR{{\mathbf{R}}}
\def\BF{{\mathbf{F}}}

\def\mW{{\mathcal{W}}}
\def\mF{{\mathcal{F}}}

\newcommand{\bbra}[1]{\left\llbracket{#1}\right\rrbracket}
\newcommand{\floor}[1]{\left\lfloor{#1}\right\rfloor}

\newcommand*\mean[1]{\overline{#1}}
\begin{document}

\begin{frontmatter}
\title{Holographic Sensing}

\author[freddy,ntu]{A.~M.~Bruckstein}
\ead{ambruckstein@ntu.edu.sg, freddy@cs.technion.ac.il}

\cortext[cor1]{Corresponding author}
\author[ntu]{M.~F.~Ezerman\corref{cor1}}
\ead{fredezerman@ntu.edu.sg}

\author[ntu]{A.~A.~Fahreza}
\ead{adamas@ntu.edu.sg}

\author[ntu]{S.~Ling}
\ead{lingsan@ntu.edu.sg}

\address[freddy]{Department of Computer Science, Technion, Israel Institute of Technology, Haifa 32000, Israel.}

\address[ntu]{School of Physical and Mathematical Sciences, Nanyang Technological University, 21 Nanyang Link, Singapore 637371.}

\begin{abstract}
Holographic representations of data encode information in packets of equal importance that enable progressive recovery. The quality of recovered data improves as more and more packets become available. This progressive recovery of the information is independent of the order in which packets become available. Such representations are ideally suited for distributed storage and for the transmission of data packets over networks with unpredictable delays and or erasures.

Several methods for holographic representations of signals and images have been proposed over the years and multiple description information theory also deals with such representations. Surprisingly, however, these methods had not been considered in the classical framework of optimal least-squares estimation theory, until very recently. We develop a least-squares approach to the design of holographic representation for stochastic data vectors, relying on the framework widely used in modeling signals and images.
\end{abstract}


\end{frontmatter}

%\linenumbers

\section{Introduction}\label{sec:intro}

Reducing the dimension of data in manners that preserve some important properties or guarantee a desired level of recovery, despite the presence of noise, has been a recurring theme of research in data processing. Examples of prominent techniques include {\it successive refinement} of information, {\it compressive} (or {\it compressed}) {\it sensing}, and {\it multiple description coding}. 

One may want to optimally describe a data given a particular level of distortion before  deciding, later on, that the data needs to be described more accurately. This naturally leads to the need for a successive refinement of information. The goal is to achieve an optimal description \emph{at each stage} as more and more information is supplied. Equitz and Cover provides a characterization of such problems from rate-distortion theory in~\cite{EC91}. They discuss two major tasks. The first is to determine the minimum rate at which information about the source must be conveyed to the user in order to achieve a given level of fidelity. The second is to investigate channels that have the minimum capacity to convey the information for a prescribed distortion. Their work is the basis of many follow-up inquiries. 

Compressive sensing simultaneously senses and compresses a signal that, under sparsity conditions, retains complete information on the data. In the sensing process, the signal is projected onto a set of vectors, which can be specifically designed or randomly chosen. The recovery process is subsequently performed by solving an inverse problem. Several seminal papers, \eg, the works of Donoho~\cite{Donoho06} and Cand{\`e}s, Romberg and Tao~\cite{CandesTR2006} set up a strong theoretical foundation for compressed sensing. Since then researchers have come up with more detailed analyses and algorithms based on various practical models with accompanying constraints and optimization objectives. The work of Elad~\cite{Elad07} is an early example that provides significant improvement over the random projection model. Many other approaches can already be found in textbooks, such as~\cite{EK12}. More recent refinements include the {\it adaptive} model where the measurement, \ie, the projection matrix, is adaptively designed using either prior information on the sparse signal or from previous measurements. Another common thread (see, \eg, the discussion in~\cite{WWS14}) is the design of some linear compression matrix that minimizes the mean squared error (MSE) or maximizes the information rate at the optimal compression ratio under some bandwidth limitation.

Multiple description coding (MDC) (see, \eg, the exposition of Goyal in~\cite{Goyal01}) is motivated by the need to reduce our dependence on the delivery mechanism where the {\it ordering} of the data packets is crucial. Its design philosophy assumes that the transport mechanism, \ie, the modulation, channel coding, and transmission protocol, is somewhat flawed or unpredictable. Hence, it is imperative to ensure that the usefulness of the bits that do arrive is more important than how many bits are available. A notable extension of MDC is the use of wavelet for image coding treated by Servetto \etal in~\cite{Servetto2000}.
		
In this work we focus on {\it holographic sensing} where information is encoded in packets of equal importance, enabling progressive recovery. As more and more packets become available, the recovered data improves progressively. The quality of this improvement must remain independent of the order in which packets become available. Several methods for holographic representations of signals and images have earlier been proposed, \eg, in~\cite{Bruckstein1998}. 
%A careful analysis of these methods in the classical framework of optimal least squares estimation theory remains largely unavailable. 
We develop a {\it least-squares approach} to the design of holographic representation for stochastic data vectors using the framework widely used in modeling signals and images. The design criteria emphasizes {\it smoothness}, 
an important aspect that has often been overlooked. Such representations are ideally suited for distributed storage and transmission or communication of data packets over networks with unpredictable delays or erasures.
 	
We start by fixing some notations in the rest of this introduction. Section~\ref{sec:prelim} explains our objectives and design philosophy by way of a toy example. Sections~\ref{sec:align} and~\ref{sec:unaligned} discuss, respectively, the situations for stochastic data vectors under the assumption that the projections are either aligned or unaligned with the standard representation basis. The treatment for the cyclostationary data vectors is given in Section~\ref{sec:cyclo}. 
Section~\ref{sec:comp} details computational implementations. Some examples in various scenarios highlight insights gleaned from actual input parameters. Section~\ref{sec:Kuty} compares and contrasts our design with that of Kutyniok \etal in~\cite{Kutyniok2009}. Their method, based on the {\it Grassmannian packing} and the theory of {\it frames}, was an initial inspiration in our investigation. % and, hence, warrants a special treatment. 
Section~\ref{sec:conclusion} concludes this work with a brief summary and a list of further directions to pursue. 

Let $0 \leq k < \ell $ be integers. Denote by $\bbra{\ell}$ the set $\{1,2,\ldots,\ell\}$ and by $\bbra{k,\ell}$ the set $\{k,k+1,\ldots,\ell\}$. Let $\bN,\bR$, and $\bC$ denote, respectively, the set of positive integers, the field of real numbers, and the field of complex numbers. The conjugate of $c \in \bC$ is denoted by ${c^*}$. Vectors are expressed as columns and denoted by bold lowercase letters. Matrices are represented by either bold uppercase letters or upper Greek symbols. An $n \times n$ diagonal matrix with diagonal entries $v_j: j \in \bbra{n}$ is denoted by $\diag(v_1,v_2,\ldots,v_n)$. The identity matrix is $\BI$ or $\BI_n$ if the dimension $n$ is important. Concatenation of vectors or matrices is signified by the symbol $|$ between the components. The transpose and the conjugate transpose of a matrix $\BA$ are $\BA^{\top}$ and $\BA^{\dagger}$, respectively.

\section{Preliminaries}\label{sec:prelim}

Audio and video signals as well as still images and a wealth of other spatio-temporally indexed data are effectively encoded in high dimensional vectors. They may be regarded as realizations of a stochastic process $\{\bx_{\omega}: \omega \in \Omega\}$ for some index set $\Omega$ where $\omega$ denotes the random choice of a particular realization and $\bx_{\omega} \in \bR^{M}$ with $M$ being the (often very high) dimension of the signal space. A classical way to characterize the properties of the process is via ensemble averages. Here the first two moments, namely the {\it mean} and the {\it autocovariance}, are of particular interest and importance. Letting $\EE_{\{\omega \in \Omega\}}$ to be the ensemble averaging operator, the mean is $\mean{\bx}=\EE_{\{\omega \in \Omega\}}[\bx_{\omega}]$ and the autocovariance is $\BR_{xx}=\EE_{\{\omega \in \Omega\}}[\bx_{\omega} \bx_{\omega}^{\top}]$. When there is no confusion, we use $\EE$ or $\EE_{\{\omega\}}$ instead of $\EE_{\{\omega \in \Omega\}}$.

We often center the data to have $\mean{\bx}=\0$ and, hence, the $M \times M$ autocovariance matrix $\BR_{xx}$ displays the variances of the entries of $\bx_{\omega}$ and the possible covariances between them. It is well-known that $\BR_{xx}$ is symmetric positive definite with a spectral decomposition
\begin{equation}\label{eq:spectral}
\BR_{xx}=\Psi \Lambda \Psi^{\dagger} \mbox{ where } \Lambda =\diag(\lambda_1,\lambda_2,\ldots,\lambda_M) \mbox{, with } \lambda_1 \geq \lambda_2 \geq \ldots\geq \lambda_{M}
\end{equation}
that displays the ordered eigenvalues $\lambda_j$s and the columns of $\Psi$ are their corresponding eigenvectors.

This work assumes that a vector $\bx \in \bR^{M}$ is a realization of a random process with zero mean and a given autocovariance matrix $\BR_{xx}$. For representation purposes $\bx$ will be projected into subspaces of $\bR^{M}$ of dimension $m << M$. It is also assumed that there is an error associated with these projections that can be modelled as an additive noise. The noise vectors are also realizations of a stochastic process $\{\bn_{\widetilde{\omega}}\}$ with zero mean and autocovariance $\BR_{nn} \triangleq \sigma_n^2 \BI_m$. Hence, we assume that the noise process has independent identically distributed entries of variance $\sigma_n^2$.

Suppose that the orthogonal projection operator $\BP_{w}$ projects vectors from $\bR^M$ onto a subspace of dimension $m$. If $\BU_w$ is an $M \times m$ matrix whose columns form an orthonormal basis for $\BP_{w}$, we have $\BP_{w} = \BU_w \BU_w^{\top}$ and the operation $\BU_w^{\top} \bx$ produces a vector of $m$ entries displaying the coefficients of the representation of $\BP_{w} \bx$ in the basis represented in $\BU_w$. Indeed, $\BP_{w} \bx = \BU_w (\BU_w^{\top} \bx)$. We probe the vector $\bx$ by measuring the vector $\BU_w^{\top} \bx$ of coefficients of $\BP_{w} \bx$ to construct the data packet which is a column vector of length $m$ given by $\displaystyle{\bz \triangleq \BU_w^{\top} \bx + \bn}$ 
where $\bn$, independent of $\bx$, is the above-mentioned realization of a white noise process with zero mean and covariance $\sigma_n^2 \BI_m$.

The classical theory of Wiener filtering (see, \eg, \cite[Chapter~3]{KSH00}) provides us with the following result. Given the data $\bz$ and the matrix $\BU_w^{\top}$ and the second order statistics of $\bn$ and $\bx$, namely $\BR_{nn}=\sigma_n^2 \BI$ and $\BR_{xx}=\EE[\bx_{\omega} \bx_{\omega}^{\top}]$, the {\em optimal estimator} for $\bx$ in the {\it expected mean squared error} sense is $\displaystyle{\widehat{\bx} = \BR_{xz} \BR_{zz}^{-1} \bz}$ with error $\be \triangleq \bx - \widehat{\bx}$ of covariance $\displaystyle{ \BR_{ee}=\EE_{\{\omega,\widetilde{\omega}\}}[\be \be^{\top}]=\BR_{xx} - \BR_{xz} \BR_{zz}^{-1} \BR_{zx}}$. Here, using
\begin{align*}
\BR_{xz}=\EE[\bx \bz^{\top}] &=\BR_{xx} \BU_w \mbox{ of size } M \times m,\\
\BR_{zx}=\EE[\bz \bx^{\top}] &=\BU_w^{\top} \BR_{xx} \mbox{ of size } m \times M,\\
\BR_{zz}=\EE[\bz \bz^{\top}] &= \BU_w^{\top} \BR_{xx} \BU_w + \sigma_n^2 \BI_m \mbox{ of size } m \times m,
\end{align*}
one derives 
\begin{equation}\label{eq:optest}
\BR_{ee}=\BR_{xx} - \BR_{xx} \BU_{w} \left(\BU_w^{\top} \BR_{xx} \BU_w + \sigma_n^2 \BI_m\right)^{-1} \BU_{w}^{\top} \BR_{xx}=
\left(\BR_{xx}^{-1}+\frac{1}{\sigma_n^2} \BU_{w}\BU_{w}^{\top}\right)^{-1}.
\end{equation}
The second equality comes from the Sherman-Morrison-Woodbury Formula, given below, for matrix inversion with $\BA=\BR_{xx}^{-1}$, $\BC=\sigma_n^{-2} \BU_{w}$, and $\BD=\BU_{w}^{\top}$. We assume $\sigma_n^2 >0$ since, otherwise, we are in the noiseless case, which can easily be treated separately.

\begin{proposition}\cite[p.~65]{Golub2013}(Sherman-Morrison-Woodbury Formula)\label{propW} Given an $n \times n$ invertible matrix $\BA$, an $n \times k$ matrix $\BC$, and a $k \times n$ matrix $\BD$, let $\BB=\BA+\BC \BD$. Let $(\BI_{k}+ \BD \BA^{-1} \BC)$ be invertible. Then 
$\BB^{-1}= \BA^{-1} - \BA^{-1} \BC (\BI_{k}+ \BD \BA^{-1}  \BC)^{-1} \BD \BA^{-1}$. 
\end{proposition}
Given a single projection operator $\BP_{w}=\BU_w \BU_w^{\top}$, the matrix $\BR_{ee}$ in (\ref{eq:optest}) can be written as $\displaystyle{\BR_{ee}=
\left(\BR_{xx}^{-1}+\frac{1}{\sigma_n^2} \BP_{w}\right)^{-1}}$. We consider the following interesting cases.  
\begin{enumerate}
	\item $\BR_{xx}=\lambda \BI_M$ for a given $\lambda >0$.
	\item $\BR_{xx}=\Lambda \triangleq \diag(\lambda_1,\lambda_2,\ldots, \lambda_M)$.
	\item $\BR_{xx}=\Psi \Lambda \Psi^{\dagger}$ with $\Psi$ a unitary or orthogonal matrix, \ie, $\Psi^{-1}=\Psi^{\dagger}$ or $\Psi^{-1}=\Psi^{\top}$.
\end{enumerate} 
Let the chosen orthonormal basis for $\bR^{M}$ be the natural basis $\{\bb_j\}_{j \in \bbra{M}}$ with $\bb_j$ the vector $(0,\ldots,0,1,0,\ldots,0)^{\top}$ having $1$ in the $j$-th position. The projection operators select $m$ samples from the vector $\bx$, \ie, $\BU_w=(\bb_{k_1}|\bb_{k_2}|\ldots|\bb_{k_m})$, implying that $\BP_{w}=\BU_{w} \BU_w^{\top}$ is a diagonal matrix with entries $1$ at locations $k_1,k_2,\ldots,k_m$ and $0$ elsewhere. 

\begin{proposition}\label{prop:basic}
Let $\{\bb_j\}_{j \in \bbra{M}}$ be the orthonormal natural basis. Using the projection operator $\BP_{w}=\BU_{w} \BU_w^{\top}$ with $\BU_w=(\bb_{k_1}|\bb_{k_2}|\ldots|\bb_{k_m})$ we obtain the following results.
\begin{enumerate}
\item If $\BR_{xx}=\lambda \BI_M$ for a given $\lambda >0$, then 
\[
\MSE=\Tr(\BR_{ee})=M \lambda - \frac{m \lambda^2}{\sigma_n^2+\lambda}
=\lambda \left(M - \frac{m}{1+\frac{\sigma_n^2}{\lambda}}\right).
\]

\item If $\BR_{xx}=\Lambda$, then
\begin{equation}\label{eq:case2}
\MSE=\Tr(\BR_{ee})= \sum_{\ell=1}^M \lambda_{\ell} - \sum_{j=1}^m 
\frac{\lambda_{k_j}^2}{\sigma_n^2+\lambda_{k_j}}.
\end{equation}
\end{enumerate}
\end{proposition}

\begin{proof}
In all cases, we use (\ref{eq:optest}) to compute for $\BR_{ee}$. 

Let $\BR_{xx}=\lambda \BI_M$. Note that  $\BR_{ee}=\displaystyle{\left(\lambda^{-1} \BI_M + \frac{1}{\sigma_n^2} \BP_{w}\right)^{-1}}$ is diagonal with positive entries
\[
\beta_j=
\begin{cases}
\lambda \mbox{ if } j \notin \{k_1,k_2,\ldots,k_m\},\\
\left(\frac{1}{\lambda}+\frac{1}{\sigma_n^2}\right)^{-1}=\frac{\lambda \sigma_n^2}{\sigma_n^2+\lambda} \mbox{ if } j \in \{k_1,k_2,\ldots,k_m\}.
\end{cases}
\] 
Hence, $\Tr(\BR_{ee})$ is given by 
\[
\lambda(M-m)+ m \frac{\lambda \sigma_n^2}{\sigma_n^2+\lambda}=\lambda M + m \left(\frac{\lambda \sigma_n^2}{\sigma_n^2+\lambda} - \frac{\lambda(\sigma_n^2+\lambda)}{\sigma_n^2+\lambda}\right)=\lambda M - \frac{\lambda^2 m}{\sigma_n^2+\lambda}.
\]
	
If $\BR_{xx}=\Lambda$, then $\BR_{ee}$ is a diagonal matrix with positive entries
\[
\beta_{\ell}=
\begin{cases}
\lambda_{\ell} \mbox{ if } \ell \notin \{k_1,k_2,\ldots,k_m\},\\
\left(\frac{1}{\lambda_{\ell}}+\frac{1}{\sigma_n^2}\right)^{-1}=
\frac{\lambda_{\ell} \sigma_n^2}{\sigma_n^2+\lambda_{\ell}} \mbox{ if } \ell \in \{k_1,k_2,\ldots,k_m\}.
\end{cases}
\] 
Hence, $\Tr(\BR_{ee})=\displaystyle{\left(  \sum_{\ell=1}^{M} \lambda_{\ell}-\sum_{j=1}^m \lambda_{k_j} \right)
+ \sum_{j=1}^m \frac{\lambda_{k_j} \sigma_n^2}{\sigma_n^2+\lambda_{k_j}}
=\sum_{\ell=1}^{M} \lambda_{\ell}-\sum_{j=1}^m \frac{\lambda_{k_j}^2}{\sigma_n^2+\lambda_{k_j}}}$.
\end{proof}

The situation is more complicated if $\BP_{w}$ is any orthogonal projection operator, \ie, $\BP_{w} = \BU_{w} \BU_w^{\top}$ where $\BU_{w}$ is a known but otherwise \emph{arbitrary} left-orthogonal basis for the subspace onto which $\BP_{w}$ projects. Note, however, that if we project onto a subspace with basis vectors given by the columns of the matrix $\Psi \BU_{w}$, which is a matrix ``adapted'' via $\Psi$ to the statistics of the $\bx$-process, we obtain 
$\displaystyle{\widetilde{\BP}_{w} \triangleq \Psi \BU_{w} \BU_{w}^{\top} \Psi^{\top}=\Psi \BP_{w} \Psi^{\top}}$ where $\BP_{w}$ is now, again, the diagonal matrix with entries $1$ at locations $k_j: j \in \bbra{m}$ and $0$ elsewhere.

\begin{proposition}\label{prop:gencase}
In the general case of $\BR_{xx}=\Psi \Lambda \Psi^{\top}$, with $\widetilde{\BP}_{w}$ and 
$\displaystyle{\widetilde{\bz} \triangleq \BU_{w}^{\top} \Psi^{\top} \bx 
+ \bn}$, we estimate the vector $\bx$ via $\widehat{\bx}=\BR_{xx} \BR_{\widetilde{z} \widetilde{z}} \widetilde{\bz}$. The $\MSE$ is given in (\ref{eq:case2}).
\end{proposition}

\begin{proof}
The resulting optimal error covariance matrix is 
\begin{align*}
\widetilde{\BR}_{ee} &\triangleq 
\left(\left(\Psi \Lambda \Psi^{\top}\right)^{-1}+\frac{1}{\sigma_n^2} \widetilde{\BP}_{w}\right)^{-1}=
\left(\Psi \left(\Lambda^{-1} +\frac{1}{\sigma_n^2} \BP_{w} \right) \Psi^{\top}\right)^{-1}\\
&=\Psi \left(\Lambda^{-1} +\frac{1}{\sigma_n^2} \BP_{w}\right)^{-1} \Psi^{\top}.
\end{align*}
The fact that the trace mapping is linear and invariant under cyclic permutations implies that $\MSE=\Tr(\widetilde{\BR}_{ee})$ is the one already derived in~(\ref{eq:case2}).
\end{proof}

To design the holographic representations we use the types of probings of the vector $\bx \in \bR^{M}$ described above. The vector is a realization of a random process $\{\bx_{\omega}: \omega \in \Omega\}$ with known statistics $\displaystyle{\EE_{\{\omega\}}[\bx_{\omega}]=\0}$ and 
$\displaystyle{ \EE_{\{\omega\}}[\bx_{\omega} \bx_{\omega}^{\top}]=\BR_{xx}}$. Probings are done via orthogonal projections onto subspaces of $\bR^{M}$. The measurements are in general contaminated by noise vectors that are independent of $\bx$ with independent and identically distributed (i.i.d.) entries of mean $0$ and variance $\sigma_n^2$. We aim for arrangements of subspaces that yield ``equally important" projections in the sense of \emph{providing similar information} about $\bx$. These projections must combine in the process of estimating $\bx$ in such a way that any pair, any triplet, and more generally any $\ell$-tuple of them yield similar restoration quality in their estimation of $\bx$. Furthermore, as the number of projections increases, the quality of the recovery should improve to a level that reaches the best possible, given the amount of data that has been made available up to that point. The holographic representation property ensures that the quality of estimating $\bx$ depends only on the number of probing data packets available, independent of the specific projections onvolved.

To set the stage, consider $\BR_{xx}=\lambda \BI_M$, \ie, the data is a vector with uncorrelated entries having variances all equal to $\lambda$. Assume further that $M = N \cdot m $. It is immediate to propose the design of $N$ subspaces of $\bR^M$, each of dimension $m$, having orthonormal bases selected from the set $\{\bb_1,\bb_2,\ldots,\bb_M\}$ such that no $\bb_j$ appears in two distinct subspace bases. This yields a set of $N$ subspaces $\{\mW_{1},\mW_{2},\ldots, \mW_{N}\}$ so that the corresponding projection operators $\BP_{1},\BP_{2},\ldots,\BP_{N}$ are diagonal with $m$ ones in locations that are pairwise disjoint and $\sum_{j=1}^{N} \BP_{j}=\BI_M$. In the language of fusion frames (see, \eg, \cite[Sect. 1.3]{CK12}) we form a rather trivial {\it Parseval fusion frame}. 

\begin{definition}\label{def:FF}
A fusion frame for $\bR^{M}$ is a finite collection of subspaces $\{\mW_j\}_{j=1}^{N}$ in $\bR^{M}$ such that, for any $\bx \in \bR^{M}$, there exist constants $0 < A \leq B < \infty$ satisfying
\begin{equation}\label{eq:frame}
A \norm{\bx}^{2} \leq \sum_{j=1}^N \norm{\BP_j \bx}^2 \leq B \norm{\bx}^2, \ie, A \BI \leq \sum_{j=1}^{N} \BP_j \leq B \BI.
\end{equation}
It is tight if $A=B$ and a tight fusion frame is a Parseval frame when $A=1$. Here $\norm{\bx}$ denotes the length or the modulus of $\bx$ and matrix inequality is defined according to the entries in their corresponding positions.
\end{definition} 

\begin{wrapfigure}{r}{0.48\textwidth}
\vspace{-0.8cm}
\begin{center}
\includegraphics[width=0.45\textwidth]{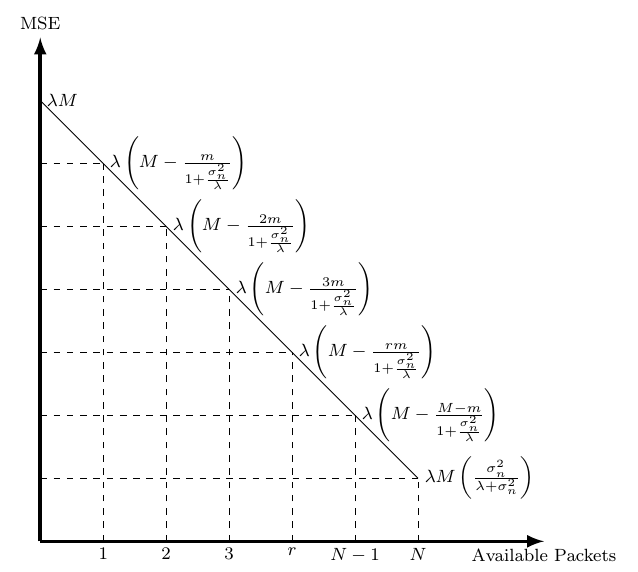}
\end{center}
\vspace{-0.4cm}
\caption{The $\MSE$ Curve for the Toy Example}\label{fig1}
\vspace{-0.8cm}
\end{wrapfigure}

In the case discussed, each data packet $\bz$ provides information on $\bx$, giving estimates for $m$ entries in $\bx$. From Proposition~\ref{prop:basic}, the optimal mean squared error of estimating $\bx$ from a single frame is
$\displaystyle{\MSE_{1\mbox{ packet}}=\lambda \left(M-\frac{m}{1+\frac{\sigma_n^2}{\lambda}}\right)}$. 
Getting $r$ pieces of data, \ie, some $\bz_{k_1},\bz_{k_2},\ldots,\bz_{k_r}$ means having a bigger projection subspace of dimension $r \cdot m$, yielding 
\[
\MSE_{r \mbox{ packets}}=\lambda \left(M-\frac{r \cdot m} {1+\frac{\sigma_n^2}{\lambda}}\right).
\]
The availability of all $N$ packets results in estimating $\bx$ with mean squared error
\[
\MSE_{\mbox{all packets}}=\lambda \left(M-\frac{N \cdot m}{1+\frac{\sigma_n^2}{\lambda}}\right)= \frac{M \sigma_n^2}{1+\frac{\sigma_n^2}{\lambda}}
\] 
with perfect recovery of $\bx$ as $\sigma_n^2 \to 0$. We have achieved our dream of having a perfect solution with a holographic representation that satisfies all of our requirements. The data packets are $\{\bz_j : j \in \bbra{N}\}$ and their performance is ideal. The best estimate of $\bx$ is reached when it is probed with all $N$ projections, \ie, when all $N$ packets are available. Figure~\ref{fig1} shows that any data set of $r \in \bbra{N}$ packets yield the same $\MSE$. 

The case we have just analysed, albeit being trivial, explains our aim clearly. The general case, when subspaces of the projections intersect and their bases do not necessarily align with the standard basis for the data vector, poses several interesting challenges. 

Our general design philosophy in allocating the subspaces is as follows. First, we want the subspace arrangements that produce the best possible $\MSE$ when all $N$ packets are available. Among the candidates satisfying this requirement we select one that has an overall smoothness property in the recovery when the number of available measurement packets is between $1$ and $N-1$. Smoothness is computed based on the relative variances of the $\MSE$ reductions, given any $\ell$ packets selected from all of the projections. We will discuss the numerical methods to come up with suitable choices below. 

\section{The Aligned Case}\label{sec:align}

This section considers the three cases of $\BR_{xx}$ when the projections on intersecting subspaces have bases that are still aligned with the standard basis representation of $\bR^M$ for $\bx$. If several data packets are available, we can apply the Wiener filter and then compute the general formula for the error from the observation
\[
\bz_{\mbox{combi}}=
\underbrace{\begin{pmatrix}
	\bz_{k_1}\\
	\bz_{k_2}\\
	\ldots\\
	\bz_{k_{\ell}}
	\end{pmatrix}}_{(\ell \cdot m) \times 1}
=\underbrace{\begin{pmatrix}
	\BU_{k_1}^{\top}\\
	\BU_{k_2}^{\top}\\
	\ldots\\
	\BU_{k_{\ell}}^{\top}\\
	\end{pmatrix}}_{(\ell \cdot m) \times M} \bx + 
\underbrace{\begin{pmatrix}
	\bn_{k_1}\\
	\bn_{k_2}\\
	\ldots\\
	\bn_{k_{\ell}}\\
	\end{pmatrix}}_{(\ell \cdot m) \times 1}.
\]
The white noise is by assumption i.i.d. with variance $\sigma_n^2 \BI_{(\ell \cdot m)}$. The $M \times M$ {\it combined} projection matrix is 
\[
\BP_{\mbox{combi}}=
\underbrace{\left(\BU_{k_1}|\BU_{k_2}|\ldots|\BU_{k_{\ell}}\right)}_{\triangleq \BU}
\begin{pmatrix}
\BU_{k_1}^{\top}\\
\BU_{k_2}^{\top}\\
\vdots\\
\BU_{k_{\ell}}^{\top}\\
\end{pmatrix}
=\sum_{j =1}^{\ell} \BU_{k_j} \BU_{k_j}^{\top}=
\sum_{j=1}^{\ell} \BP_{k_j},
\] 
yielding the error covariance matrix
\begin{align}\label{eq:A_Ree}
\BR_{ee}&=\left(\BR_{xx}^{-1}+\frac{1}{\sigma_n^2} \sum_{j=1}^{\ell} \BP_{k_j}\right)^{-1}=\left(\BR_{xx}^{-1}+\frac{1}{\sigma_n^2} \BU \BU^{\top}\right)^{-1} \notag \\ 
&=
\BR_{xx} - \BR_{xx} \frac{1}{\sigma_n^2} \BU \left(\BI + \BU^{\top} \BR_{xx} \frac{1}{\sigma_n^2} \BU\right)^{-1} \BU^{\top} \BR_{xx}.
\end{align}

Let $\mW_{k_j}$ be the subspace onto which $\BP_{k_j}$ projects. Suppose that $\mW_{k_j}$ for each $j \in \bbra{\ell}$ is aligned with the standard basis $\{\bb_1,\bb_2,\ldots,\bb_M\}$, \ie, $\BP_{k_j}$ is a diagonal matrix with diagonal entries $1$ or $0$ corresponding, respectively, to whether a certain coordinate of $\bx$ is probed or not. This implies that $\sum_{j=1}^k \BP_{k_j}$ is also diagonal with nonnegative integer diagonal entries displaying how often a certain coordinate of $\bx$ was probed. Then the $\BR_{ee}$ has a pleasingly simple formula for its trace that gives the expected $\MSE$ from the Wiener filter recovery. Let $\cP_s$ for $s \in \bbra{0,\ell}$ be the set of positions in the diagonal of $\BP_{\mbox{combi}}$ whose entries are $s$. Note that $\sum_{s=0}^{\ell} \abs{\cP_s}= M$ and $\sum_{s=1}^{\ell} s \abs{\cP_s}= \ell \cdot m$.

\vspace{0.3cm}
\noindent {\bf Case 1}: 
Let $\BR_{xx}=\lambda \BI_M$. Assume that there are $\ell \in \bbra{N}$ arbitrary measurement packets available to approximate $\bx$. Then $\displaystyle{
	\BR_{ee}=\left(\BR_{xx}^{-1}+\frac{1}{\sigma_n^2} \sum_{j=1}^{\ell} \BP_{k_j}\right)^{-1}}$ is diagonal with positive entries $\alpha_j=\displaystyle{\frac{\lambda \sigma_n^2}{\sigma_n^2+s \lambda}}$ for $j \in \cP_s$ with $s \in \bbra{0,\ell}$. Hence, $\Tr(\BR_{ee})$ is 
\begin{multline}\label{eq:MSElambda}
\MSE(\lambda \BI_{M},\sigma_n^2,\ell)=\sum_{s=0}^{\ell} \sum_{j \in \cP_s} \frac{\lambda \sigma_n^2}{\sigma_n^2+s \lambda}=
\sum_{s=0}^{\ell} \abs{\cP_s} \lambda \left(1-\frac{s \lambda}{\sigma_n^2+ s \lambda}\right)\\
= \sum_{s=0}^{\ell} \abs{\cP_s} \lambda - \sum_{s=1}^{\ell} \abs{\cP_s} 
\frac{s \lambda^2}{\sigma_n^2 + s \lambda} 
= M \lambda -\sum_{s=1}^{\ell} \frac{\abs{\cP_s} \lambda}{1+\frac{\sigma_n^2}{s \lambda}}. 
\end{multline}

\begin{remark}
In the toy example of Section~\ref{sec:prelim}, $s\in \{1\}$ with $\abs{\cP_1}=\frac{\ell \cdot M}{N}=\ell \cdot m$ is the only possibility. Hence, $\MSE(\lambda,\sigma_n^2,\ell)=\lambda M 
\displaystyle{\left(1 - \frac{\lambda \ell}{N (\lambda+\sigma_n^2)}\right)}$, as had been shown. 
\end{remark}

\vspace{0.3cm}
\noindent {\bf Case 2}: Let  $\BR_{xx}=\Lambda=\diag(\lambda_1,\lambda_2,\ldots,\lambda_M)$ and assume that all $\BP_{j}$ with $j \in \bbra{N}$ are projections onto subspaces of equal dimension $m$, \ie, $\Tr(\BP_{j})=m$. Given all $N$ packets,  
\[
\sum_{j=1}^{N} \BP_j =\diag(s_1,s_2,\ldots,s_M) \mbox{ with } 0 \leq s_j \in \bZ \mbox{ and } \sum_{j=1}^N s_j = N \cdot m.
\] 
To determine the values of $s_j$ that minimize the $\MSE$ we start from   (\ref{eq:A_Ree}) to infer that $\BR_{ee}=
\displaystyle{\left(\diag(\lambda_1^{-1},\lambda_2^{-1},\ldots,\lambda_M^{-1}) + \frac{1}{\sigma_n^2} \diag(s_1,s_2,\ldots,s_M)\right)^{-1}}$, which implies

\begin{equation}\label{eq:MSE_N}
\MSE(\Lambda,\sigma_n^2,N)=\sum_{j=1}^M \frac{\sigma_n^2 \lambda_j}{\sigma_n^2 + \lambda_j s_j}=\sum_{j=1}^M \frac{\lambda_j}{1 + \left(\frac{\lambda_j}{\sigma_n^2}\right) s_j}.
\end{equation}
% The last expression highlights the {\it noise-to-signal ratio}. 

Let us now minimize the $\MSE$ given in (\ref{eq:MSE_N}) when all probings are made available. To achieve this we solve the optimization problem $\displaystyle{\min_{\{\zeta_j\}} \MSE(\Lambda,\sigma_n^2,N)}$ using the Lagrange multipliers method, subject to $\sum_{j=1}^M \zeta_j = N \cdot m$ and $0 \leq \zeta_j \in \bR$. Let 
\[
\Theta(\zeta_1,\zeta_2,\ldots,\zeta_M) \triangleq 
\sum_{j=1}^M \frac{\lambda_j \sigma_n^2}{\sigma_n^2 + \lambda_j \zeta_j} + 
\beta \left(\sum_{j=1}^M \zeta_j - N \cdot m\right) \mbox{ with } \beta >0.
\]
Solving for $\zeta_j$ in $\displaystyle{\frac{\partial \Theta}{\partial \zeta_j}=-\frac{\lambda_j^2 \sigma_n^2}{\left(\sigma_n^2+\lambda_j \zeta_j\right)^2}+ \beta=0 \mbox{ yields }
\left(\sigma_n^2+\lambda_j \zeta_j\right)^2 =\frac{\lambda_j^2 \sigma_n^2}{\beta}}$, implying   
$\displaystyle{\zeta_j = \frac{\sigma_n}{\sqrt{\beta}}-\frac{\sigma_n^2}{\lambda_j}}$. From 
$\displaystyle{
\sum_{j =1}^M \zeta_j = M \frac{\sigma_n}{\sqrt{\beta}}-\sigma_n^2 \left(\sum_{j=1}^{M} \frac{1}{\lambda_j}\right)=N \cdot m}$ one obtains 
\[
M \frac{\sigma_n}{\sqrt{\beta}} = N \cdot m + \sigma_n^2 
\left(\sum_{j=1}^{M} \frac{1}{\lambda_j}\right) \implies 
\sqrt{\beta}= \frac{M\sigma_n}{N \cdot m + \sigma_n^2 \left(\sum_{j=1}^{M} \frac{1}{\lambda_j}\right)}.
\] 
We can then conclude that
\begin{equation}\label{eq:zetaj}
\zeta_j = \frac{\sigma_n}{M \sigma_n} \left[N \cdot m + \sigma_n^2 \left(\sum_{k=1}^{M} \frac{1}{\lambda_k}\right)\right] - \frac{\sigma_n^2}{\lambda_j}=
\frac{N \cdot m}{M}+\sigma_n^2 \left(\frac{1}{M} \sum_{k=1}^{M} \frac{1}{\lambda_k} - \frac{1}{\lambda_j} \right).
\end{equation}
The second derivative test on $\Theta$ confirms that $\zeta_j$ is indeed a local minimizer since 
$\displaystyle{
\frac{\partial^2 \Theta}{\partial \zeta_j^2}=\frac{2 \lambda_j^3 \sigma_n^2}{(\sigma_n^2+\lambda_j \zeta_j)^3} > 0}$. Thus, the optimal $\zeta_j$, in the sense of the one leading to the least $\MSE$, measures the departure of $\displaystyle{\frac{1}{\lambda_j}}$ from the {\it average} contribution $\displaystyle{\frac{1}{M}\sum_{k=1}^{M} \frac{1}{\lambda_k}}$. 
Note that in (\ref{eq:zetaj}), there may be a threshold $t$ such that $\zeta_j \geq 0$ for $j \in \bbra{t}$ and $\zeta_j <0$ for $j \in \bbra{t+1,M}$. Applying the constraint $\zeta_j \geq 0$, we set $\zeta_j=0$ for $j \in \bbra{t+1,M}$. To ensure that $\sum_{j=1}^M \zeta_j = N \cdot m$ still holds when there is such a $t$, we recompute $\sqrt{\beta}_t \triangleq 
\displaystyle{\frac{t \cdot \sigma_n}{N \cdot m + \sigma_n^2 \left( \sum_{j=1}^{t} \frac{1}{\lambda_j}\right)}}$ and use it to determine the new $\displaystyle{\zeta_j = \frac{\sigma_n}{\sqrt{\beta}_t}-\frac{\sigma_n^2}{\lambda_j}}$. The process is repeated until all $\zeta_j \geq 0$ for all $j \in \bbra{M}$. Finally, we round each $\zeta_j$ off to get $ 0 \leq s_j \in \bN$. 

After ensuring that we obtain the best possible recovery when all $N$ packets are available, we would now like to have a graceful degradation when any $\ell \in \bbra{N}$ packets, say $\bz_{k_1},\bz_{k_2},\ldots, \bz_{k_{\ell}}$, are available. Then $\displaystyle{\BR_{ee}=\left(\Lambda^{-1}+\frac{1}{\sigma_n^2} \sum_{j=1}^{\ell} \BP_{k_j}\right)^{-1}}$ is a diagonal matrix with positive entries $\displaystyle{\frac{\lambda_j \sigma_n^2}{\sigma_n^2+s \lambda_j}}$ for $j \in \cP_s$ with $s \in \bbra{0,\ell}$. Taking the trace yields
\begin{align}\label{eq:l_align}
\MSE(\Lambda,\sigma_n^2,\ell)&=\sum_{s=0}^{\ell} \sum_{j \in \cP_s} \frac{\lambda_j \sigma_n^2}{\sigma_n^2+s \lambda_j}
= \sum_{j \in \cP_0} \lambda_j + \sum_{j \notin \cP_0} \lambda_j - \sum_{j \notin \cP_0} \lambda_j  +
\sum_{s=1}^{\ell} \sum_{j \in \cP_s} \frac{\lambda_j \sigma_n^2}{\sigma_n^2+s \lambda_j} \notag \\
&=\sum_{j=1}^M \lambda_j +\sum_{s=1}^{\ell} \sum_{j \in \cP_s} 
\left(\frac{\lambda_j \sigma_n^2}{\sigma_n^2+s \lambda_j} - \lambda_j \right)=
\sum_{j=1}^M \lambda_j - \sum_{s=1}^{\ell} \sum_{j \in \cP_s} 
\frac{s \lambda_j^2} {\sigma_n^2 + s \lambda_j}.
\end{align}
In particular, when all $N$ packets are available, we get (\ref{eq:MSE_N}) from (\ref{eq:l_align}) since, by design, 
\begin{equation}\label{eq:N_align}
\MSE(\Lambda,\sigma_n^2,N)=\sum_{s=0}^{N} \sum_{j \in \cP_s} \frac{\lambda_j \sigma_n^2}{\sigma_n^2+s \lambda_j}=
\sum_{j=1}^M \frac{\lambda_j \sigma_n^2}{\sigma_n^2 + s_j \lambda_j}.
\end{equation}

\section{The Unaligned Case}\label{sec:unaligned}

We move to the more general setup. Let $\BP_{w}$ be any orthogonal projection operator, \ie, $\BP_{w} = \BU_{w} \BU_w^{\top}$ where $\BU_{w}$ is an arbitrary orthonormal basis for the subspace onto which $\BP_{w}$ projects. First, let $\BR_{xx}=\lambda \BI_M$. The formula when probing $\bx$ by projecting onto all subspaces described via $\BU_{1}, \BU_{2}, \ldots, \BU_{N}$ of respective dimensions $m_{j}$ for $j \in \bbra{N}$ is $\displaystyle{\BR_{ee}=\left(\frac{1}{\lambda} \BI_M + \frac{1}{\sigma_n^2} \sum_{j=1}^N \BP_{j}\right)^{-1}}$. Suppose that only a single subspace projection, say $\BP_{w}$, is made available. Then, by applying Proposition~\ref{propW} with $\BA=\lambda^{-1} \BI_{M}$, $C=\frac{1}{\sigma_n^2} \BU_{w}$, and $D=\BU_{w}^{\top}$, the error covariance matrix $\BR_{ee}$ is given by 
\[
\left(\frac{1}{\lambda} + \frac{1}{\sigma_n^2} \BP_{w}\right)^{-1}= 
\left(\frac{1}{\lambda} \BI_M + \frac{1}{\sigma_n^2} \BU_{w} \BU_{w}^{\top}\right)^{-1}=
\lambda \BI_{M} -\frac{\lambda^2}{\sigma_n^2} \BU_{w} \left(\BI_{m_w} + \frac{\lambda}{\sigma_n^2} \BI_{m_w}\right)^{-1} \BU_{w}^{\top}.
\]
Taking the trace establishes 
$\displaystyle{\MSE=\lambda M - m_w \frac{\lambda^2}{\sigma_n^2+\lambda}}$.
Thus, if we want all subspaces in the given setup to provide the same mean squared error reduction, they must have equal dimension $m_w=m$. 

If $\ell \geq 2$ probing packets of data are available, then $\displaystyle
\BR_{ee}= \left(\frac{1}{\lambda} \BI_{M}+ \frac{1}{\sigma_n^2} 
\sum_{j=1}^{\ell} \BP_{w_j} \right)^{-1}$ with $\BP_{w_j}=\BU_{w_j} \BU_{w_j}^{\top}$ being $M \times M$ matrices for all $j \in \bbra{\ell}$. 

When all $N$ packets are available we write $\displaystyle
\BR_{ee}= \left(\frac{1}{\lambda} \BI_{M}+ \frac{1}{\sigma_n^2} 
\sum_{j=1}^{N} \BP_{j} \right)^{-1}$ by letting $\widetilde{\BU} \triangleq (\BU_{1}|\BU_{2}|\ldots|\BU_{N})$. Since $\displaystyle{\sum_{j=1}^{N} \BP_{j}}$ is symmetric positive semidefinite, its trace is $N \cdot m$. In other words, if $\{\zeta_t : t \in \bbra{M}\}$ is the set of all of its eigenvalues, then $\sum_{t=1}^{M} \zeta_t= N \cdot m$ and there exists an orthogonal matrix $\Psi$ such that $\displaystyle{
\sum_{j=1}^{N} \BP_{j} = \Psi \diag(\zeta_1,\zeta_2,\ldots,\zeta_M) \Psi^{\top}}$. Hence, 
\begin{align*}
\MSE&=\Tr\left(\left(\frac{1}{\lambda} \Psi \Psi^{\top}+\frac{1}{\sigma_n^2} \Psi \diag(\zeta_1,\zeta_2,\ldots,\zeta_M) \Psi^{\top}\right)^{-1}\right)\\
&=\Tr\left(\left(\frac{1}{\lambda} \BI_{M} + \frac{1}{\sigma_n^2} \diag(\zeta_1,\zeta_2,\ldots,\zeta_M) \right)^{-1}\right)
=\sum_{t=1}^{M} \frac{\lambda \sigma_n^2}{\sigma_n^2+\lambda \zeta_t}.
\end{align*}

Solving $\displaystyle{
\min_{\{\zeta_t\}} \sigma_n^2 \sum_{t=1}^M \frac{1}{\zeta_t+\frac{\sigma_n^2}{\lambda}}}$ such that 
$\displaystyle{\sum_{t=1}^{M} \zeta_t =N \cdot m}$ gives us the minimum achievable $\MSE$. Let 
$\displaystyle{
\Upsilon(\zeta_1,\zeta_2,\ldots,\zeta_M)\triangleq \sum_{t=1}^M \frac{\lambda \sigma_n^2}{\sigma_n^2 + \lambda \zeta_t} + \alpha \left(\sum_{t=1}^M \zeta_t - N \cdot m\right)}$. Solving for $\zeta_t$ in
$\displaystyle{
\frac{\partial \Upsilon}{\partial \zeta_t} =\alpha-\frac{\lambda^2 \sigma_n^2}{(\sigma_n^2+\lambda \zeta_t)^2}=0}$ gives 
$\zeta_t=\displaystyle{\frac{\sigma_n}{\sqrt{\alpha}}-\frac{\sigma_n^2}{\lambda}}$ and $\displaystyle{
\sum_{t=1}^M \zeta_t = M \left(\frac{\sigma_n}{\sqrt{\alpha}}-\frac{\sigma_n^2}{\lambda}\right)
= N \cdot m}$ leads to $\displaystyle{
\frac{\sigma_n}{\sqrt{\alpha}}=\frac{N \cdot m}{M}-\frac{\sigma_n^2}{\lambda}}$. Hence, $\zeta_t = \frac{N \cdot m}{M}$. 
To see that this value is indeed a local minimum, notice that the second derivative $\displaystyle{
\frac{\partial^2 \Upsilon}{\partial \zeta_t^2} = \frac{2 \lambda^3 \sigma_n^2}{(\sigma_n^2+\lambda \sigma_t)^3}>0}$. Thus, to minimize the error, we need to make $\zeta_t$ as uniform as possible for all $t \in \bbra{M}$. In particular, it is desirable to have $M \mid (N \cdot m)$, \ie, to have an $A$-tight fusion frame with $A = \frac{N\cdot m}{M}$. The local minimum value for $\MSE$ is, in this case, 
\begin{equation}\label{eq:MSE_tight}
\sum_{t=1}^{M} \frac{\lambda \sigma_n^2}{\sigma_n^2+\lambda \left(\frac{N \cdot m}{M}\right)}=\frac{M^2 \lambda \sigma_n^2}{\lambda N \cdot m +  M \sigma_n^2}.
\end{equation}

In the general case where $\bx$ has $\BR_{xx}=\Psi \Lambda \Psi^{\top}$ with all $N$ packets available, we use projections of the form $\bz_j=\BU_{w_j}^{\top} \by$ with $\by \triangleq \Psi^{\top} \bx$, making $\BR_{yy}=\Lambda$. In this case, $\displaystyle{\BR_{ee}=\left(\BR_{xx}^{-1}+\frac{1}{\sigma_n^2} \sum_{j =1}^N \widetilde{\BP}_{w_j}\right)^{-1}}$ with $\displaystyle{\sum_{j =1}^N \widetilde{\BP}_{w_j}=\Psi \left(\sum_{j =1}^N \BP_{w_j}\right) \Psi^{\top}}$. Hence,
\[
\BR_{ee}=\left(\Psi \Lambda^{-1} \Psi^{\top}+ \frac{1}{\sigma_n^2} \Psi \sum_{j =1}^N \BP_{w_j} \Psi^{\top}\right)^{-1}=
\Psi \left(\Lambda^{-1}+\frac{1}{\sigma_n^2}\sum_{j =1}^N \BP_{w_j}\right)^{-1} \Psi^{\top},
\]
yielding the same $\MSE(\Lambda,\sigma_n^2,N)$ as the one already determined in (\ref{eq:N_align}). Similar reasoning yields the same formula for $\MSE(\Lambda,\sigma_n^2,\ell)$ already deduced in (\ref{eq:l_align}). 

\section{Cyclostationary Data Vectors}\label{sec:cyclo}
This section considers data whose statistical characteristics vary periodically with time. The processes that produce such data are said to be {\it cyclostationary} or {\it periodically correlated}. They are abundant in econometry, telecommunication, and astronomy. Relevant definitions, prominent examples, and further references are available in~\cite{GNP06}.

Henceforth, $i \triangleq \sqrt{-1}$ and $\omega \triangleq 
e^{- i \frac{2 \pi}{M}}$, which is a primitive $M$-th root of unity. Here we have $M$ a power of $2$ and the correlation matrix $\BR_{xx}$ is circulant with first row entries, for some $ 0 < \gamma \in \bR$:
\[
1, \gamma, \gamma^2,\ldots, \gamma^{\frac{M}{2}-1},\gamma^{\frac{M}{2}}, \gamma^{\frac{M}{2}-1},\ldots,\gamma^2,\gamma.
\]
As a consequence of the Circular Convolution Theorem from the theory of Discrete Fourier Transforms, we can write $\BR_{xx}$ as $\mF \Lambda  \mF^{\dagger}$ where $\mF$ is a (unitary) DFT matrix with entries $\mF_{j,k}=\frac{\omega^{(j-1)(k-1)}}{\sqrt{M}}$ for $j,k \in \bbra{M}$, \ie,
\begin{equation}\label{eq:Fmatrix}
\mF \triangleq \frac{1}{\sqrt{M}}
\begin{pmatrix}
1 & 1 & 1 & \ldots & 1\\
1 & \omega & \omega^2 & \ldots & \omega^{M-1}\\
1 & \omega^2 & \omega^4 & \ldots & \omega^{2(M-1)}\\
\vdots & \vdots & \vdots & \vdots & \vdots\\ 
1 & \omega^{M-1} & \omega^{2(M-1)} & \ldots & \omega^{(M-1)(M-1)}\\
\end{pmatrix},
\end{equation}
and $\Lambda$ a diagonal matrix $\diag(\lambda_1,\lambda_2,\ldots,\lambda_M)$. Note that the entries $\lambda_j$s are no longer monotonically nonincreasing. Let $\bc:=(1,\gamma,\gamma^2,\ldots,\gamma^{\frac{M}{2}-1},\gamma^{\frac{M}{2}},\gamma^{\frac{M}{2}-1},\ldots,\gamma^2,\gamma)^{\top}$ be a vector in $\bR^{M}$. We use a well-known result~\cite[Theorem 4.8.2]{Golub2013} to conclude that the diagonal entries in $\Lambda$ are the elements in vector $\mF^{\dagger} \bc$. Since $\mF$ is unitary, after some manipulation we obtain

\begin{equation}\label{eq:genlam}
\lambda_j = \frac{1}{\sqrt{M}} \left(1 + (-1)^{j-1} \gamma^{\frac{M}{2}} 
+ \sum_{k=1}^{\frac{M}{2}-1} \gamma^{k} 2 \cos \left(\frac{2 \pi k (j-1)}{M}\right) \right) \mbox{ for } j \in \bbra{M}.
\end{equation}
Hence, $\lambda_j = \lambda_{M+2-j}$ for $j \in \bbra{2,M/2}$ and $\sum_{j=1}^M \lambda_j = \sqrt{M} = \sqrt{\Tr(\BR_{xx})}$. 

\begin{example}
For $M=4$, we have $\bc^{\top}=(1,\gamma,\gamma^2,\gamma)$ and 
\[
\mF^{\dagger}= \frac{1}{2}
\begin{pmatrix}
1 & 1 & 1 & 1\\
1 & i & -1 & -i\\
1 & -1 & 1 & -1\\
1 & -i & -1 & i\\
\end{pmatrix}
\mbox{, making } 
\begin{cases}
\lambda_1=(1+2\gamma+\gamma^2)/2,\\
\lambda_2=\lambda_4=(1-\gamma^2)/2,\\
\lambda_3=(1-2\gamma+\gamma^2)/2.
\end{cases}
\]
\end{example}

The measurements are given by 
$\displaystyle{\bz_k = \BU_k^{\top} \mF^{\dagger} \bx + \widetilde{\bn}_k \triangleq \BU_k^{\top} \by + \widetilde{\bn}_k}$ for $k \in \bbra{M}$. Now that we have $\BR_{yy}=\EE[\by \by^{\top}]=\mF^{\dagger} \left( \mF \Lambda \mF^{\dagger}\right) \mF=\Lambda$, the Wiener Theory allows for the derivation of the expected error when measurements from $\ell$ arbitrary subspaces are available. As before, let $\bz_{\mbox{combi}}$ and $\BU_{\mbox{combi}}$ be the respective concatenations of available $\bz_{k_j}$ and $\BU_{k_j}$ for $j \in \bbra{\ell}$. The $\MSE$ in this case is exactly the same as the $\MSE$ in the aligned case. 
This follows since   
$\displaystyle{
\BR_{ee}=\EE \left[(\by - \BF \bz_{\mbox{combi}})(\by - \BF \bz_{\mbox{combi}})^{\dagger}\right]=\BR_{yy} -\BR_{yz} \BR_{zz}^{-1} \BR_{zy}}$
with 
\[
\BR_{yy}=\Lambda, \,
\BR_{yz}=\BR_{yy} \BU_{\mbox{combi}} =\Lambda \sum_{j \in 1}^{\ell} \BP_{k_j} \mbox{, and } 
\BR_{zz}=\BU_{\mbox{combi}}^{\top} ~\Lambda~ \BU_{\mbox{combi}} + \sigma_n^2 \BI_{(\ell \cdot m)}.
\] 
The derivation of the $\MSE$ follows the steps done in Section~\ref{sec:align}. Thus, $\MSE(\Lambda,\sigma_n^2,\ell)$ is the one given in (\ref{eq:l_align}) while $\MSE(\Lambda,\sigma_n^2,N)$ is in (\ref{eq:N_align}), with $\Lambda$ as defined in this section. To get the best approximation $\widetilde{\bx}$ of $\bx$, we apply $\mF$ on the approximation $\widetilde{\by}$ of $\by$.

\section{Computational Implementation}\label{sec:comp}

We implement the holographic sensing design computationally in a program written in {\tt python 2.7}. The program has three different modes, namely, standard, linear, and cyclostationary, in correspondence with the different models of $\Lambda$. On input $(M,m,N,\Lambda,\sigma_n^2)$ the program determines $\zeta_k \geq 0$ for $k \in \bbra{M}$ and then computes for the absolute distance of each $\zeta_k$ to the nearest integer for a proper rounding off of $\zeta_k$ to $s_k$. To ensure that $\sum_{k=1}^M s_k = N \cdot m$, there may be values of $\zeta_j$ with relatively large distance that need to be assigned to $\floor{\zeta_j}$. The program then computes for $\MSE(N)$ from (\ref{eq:MSE_N}). For a given $M$ we call the constant term $\sum_{j=1}^M \lambda_j$ in~(\ref{eq:l_align}) the {\it base point}. To highlight the {\it gain in recovery} as more packets are made available, we call $\displaystyle{\Delta(\ell) \triangleq \sum_{s=1}^{\ell} \sum_{j \in \cP_s} \frac{s \lambda_j^2} {\sigma_n^2 + s \lambda_j}}$ the $\MSE(\ell)$ {\it reduction}, which we want to maximize.

For relatively small values of $(M,m,N)$ users may choose to generate all subspace arrangements. The program comes with an option to specify a number, say $100$, of arrangements with maximal $\Delta(N)$ for each input parameter set to be uniformly generated. Two plots are produced to illustrate, respectively, the minimum $\MSE(\ell)$ and the \emph{variance} of $\Delta(\ell)$ for $\ell \in \bbra{N-1}$. The subspace arrangements are ranked from smoothest, \ie, the one with smallest normalized $\ell_2$-norm of the variances of the $\Delta(\ell)$ to the largest. The selected best arrangement is represented by the corresponding bold curves in the plots. As expected, the smoothest arrangement, while not lagging far behind, is usually not the best-performing in terms of the $\MSE$ reduction gain $\Delta(\ell)$ for each chosen $\ell$.

The {\it smoothness threshold} $\delta_{\epsilon}$ specifies the minimum number of available packets such that all subspace arrangements have variances of their $\Delta(\ell)$ reductions below $\epsilon$. If a user can tolerate $\epsilon = 0.1$, then $\delta_{0.1}$ gives the number of required packets to ensure that \emph{any arbitrarily chosen} subspace arrangement from the generated list is good enough. This works the other way as well. If at least a number of measurement packets always makes it through the channel, then one knows the variance of the $\MSE$ reductions that can be expected from using any subspace arrangement. 

The next three subsections explain how the program handles different types of data.

\subsection{A Typical Stochastic Data: $\lambda_j$ Decays Exponentially with $j$}

First, let us consider a typical stochastic data where $\Lambda=\diag(\lambda_1,\lambda_2,\ldots,\lambda_M)$ and $\lambda_j=\gamma^{j-1}$ for $0 < \gamma < 1$ and $j \in \bbra{M}$. On input $(M,m,N,\Lambda,\sigma_n^2)$ the program determines the largest positive integer $t \leq M$ such that $\zeta_k >0$ for $k \in \bbra{t}$ and then computes for the absolute distance of each $\zeta_k$ to the nearest integer for a proper rounding off of $\zeta_k$ to $s_k$, starting from the index corresponding to the lowest distance to the largest. A method to determine $t$ has been given in Section~\ref{sec:align}.

We start with a simple example. Let $M=8$, $m=4$, $N=5$, $\sigma_n^2=0.5$, and $\lambda_j=0.8^{j-1}$ for $j \in \bbra{8}$. Computation shows that $\zeta_j > 0$ for all $j$. Up to three significant figures, they are $3.24, 3.12, 2.96, 2.76, 2.52, 2.21, 1.83, 1.36$. Rounding off, we get $s_j=3$ for $j \in \bbra{5}$, $s_6=s_7=2$, and $s_8=1$. The maximal $\Delta(5)$ is $3.0864$, making $\MSE(5)= 1.075$ since the base point is $4.161$. There are $3770$ possible arrangements. The smoothest one, represented by its set of indices, is 
$\{\{1, 2, 7, 8\}, \{1, 3, 4, 7\}, \{1, 4, 5, 6\}, \{2, 3, 4, 5\}, 
\{2, 3, 5, 6\}\}$. It has normalized variance $0.0135$. As one can easily see, each of the first $5$ coordinates is sampled $3$ times, \ie, $s_j=3$ for $j \in \bbra{5}$, and so on until the last coordinate sampled only once, \ie, $s_8=1$.

Separately, we generate $300$ randomly selected arrangements having the required maximal $\Delta(5)$. In a particular run, the smoothest of these $300$, with normalized variance of $0.0147$ is 
$\{\{1, 2, 5, 8\}, \{1, 3, 4, 5\}, \{1, 4, 6, 7\}, \{2, 3, 4, 7\}, 
\{2, 3, 5, 6\}\}$. 
In both the exhaustive and random runs, if at least $3$ packets are guaranteed to be available, then \emph{any choice} of subspace arrangement has variance of $\MSE$ reductions less than $0.05$, {\it i.e.}, $\delta_{0.05}=3$. 

Figure~\ref{fig:M8} presents the respective sets of two plots, one for the exhaustive run and the other for the random run, for an easy comparison.

\begin{figure}[h!]
\hspace{-1.2cm}
\includegraphics[width=1.2\textwidth]{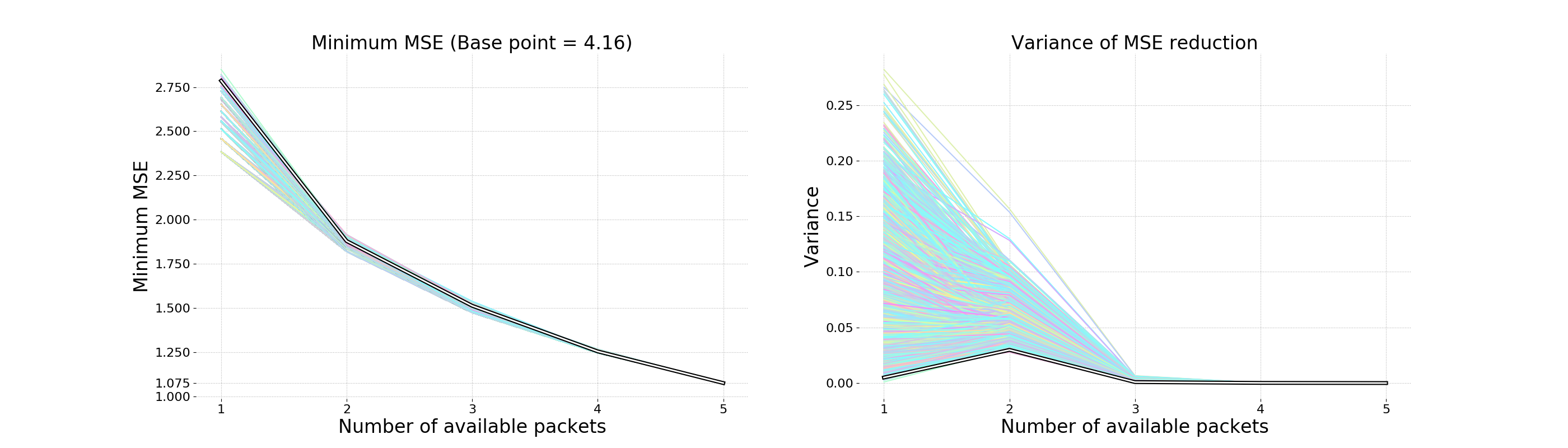}
\smallskip\par
\hspace{-1.2cm}
\includegraphics[width=1.2\textwidth]{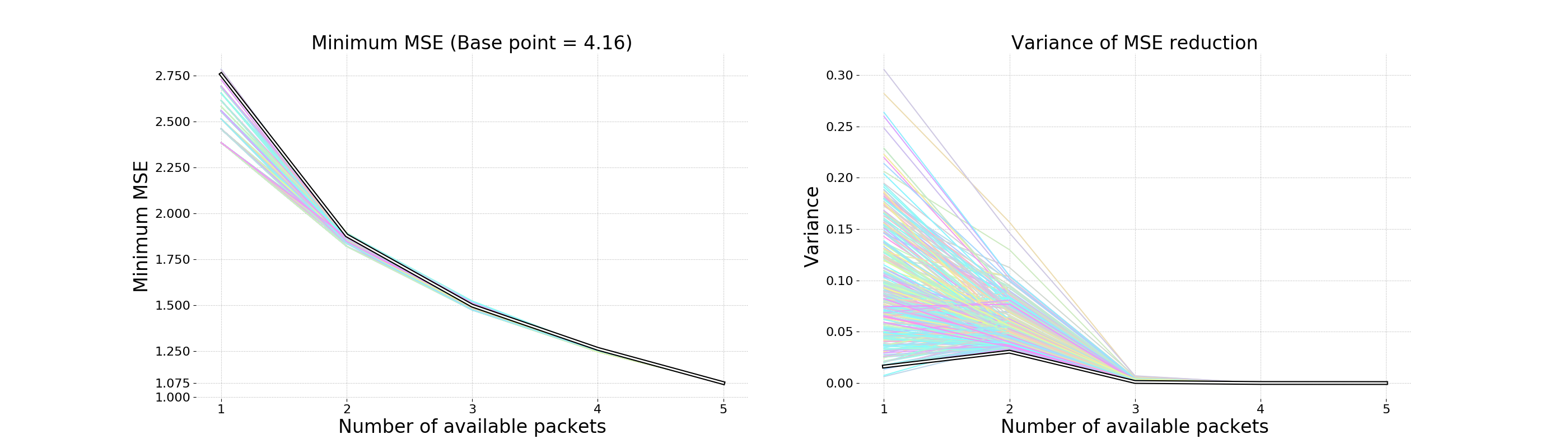}
\caption{A typical stochastic data with $M=8$, $m=4$, $N=5$, $\sigma_n^2=0.5$, and $\lambda_j=0.8^{j-1}$ for $j \in \bbra{8}$. The plot on the left depicts the $\MSE$ as a function of the number of $1 \leq \ell \leq N$ available packets for different arrangements. The base point is computed by setting $\ell=0$. The plot on the right shows the trend on the variance of the $\MSE$ reduction given the number of available packets. Above: all $3770$ arrangements. Below: uniformly selected $300$ arrangements.} \label{fig:M8}
\end{figure}

From Section~\ref{sec:align} it is clear that, regardless of the dimension $M$, given fixed $(\Lambda,\sigma_n^2)$, the values of $\MSE(N)$, $t$, and the set $\{s_j : j \in \bbra{t}\}$ depend only on $N \cdot m$. Table~\ref{table:fixed} illustrates the fact. The set $\{s_k : k \in \bbra{t}\}$ is written in shorthand with $[x_1]^{y_1}[x_2]^{y_2} \ldots [x_r]^{y_r}$ denoting $s_k=x_1$ for $k \in \bbra{y_1}$ followed by $s_k=x_2$ for $k \in \bbra{y_1+1,y_1+y_2}$ and so on until $s_k=x_r$ for $k \in \bbra{t-y_r+1,t}$. We remove the superscript if it is $1$. For example, Entry 1 in Table~\ref{table:fixed} has $[3][2]^{13}[1]^3$ in the specified column of $\{s_k : k \in \bbra{t}\}$ with $t=17$. This means that $s_1=3$, $s_k=2$ for $ 2 \leq k \leq 14$ and $s_k=1$ for $15 \leq k \leq 17$. 

\begin{table}[h!]
\caption{Computed values for typical stochastic data with $\lambda_j=0.8^{j-1}~:~j \in \bbra{M}$ and $\sigma_n^2=0.05$}
\label{table:fixed}
\setlength{\tabcolsep}{0.1cm}
\centering
\begin{tabular}{c c c c c c c c l c}
	\hline
	No. & $M$ & $N$ & $m$ & $\sum_{j=1}^M \lambda_j$ & Max $\Delta(N)$ & $\MSE(N)$ & $t$ & $\{s_k : k \in \bbra{t}\}$ & $\delta_{0.1}$ \\
	\hline
	1 & $64$ & $8$ & $4$ & $5.00$ & $4.53$ & $0.47$ & $17$ & $[3][2]^{13}[1]^3$ & $6$ \\
		
	2 &  	& $8$ & $8$ & &  $4.69$ & $0.31$ & $19$ & $[4]^{11}[3]^5 [2]^2 [1]$ & $4$ \\ 
		
	3 &     & $16$ & $4$ & & $4.69$ & $0.31$ & $19$ & $[4]^{11}[3]^5 [2]^2 [1]$ & $9$ \\
				
	4 & $1024$  & $8$ & $8$ &  & $4.69$ & $0.31$ & $19$ & $[4]^{11}[3]^5 [2]^2 [1]$ & $4$ \\
		
	5 & & $16$&$8$ & & $4.81$ & $0.19$ & $22$ & $[7]^{10}[6]^6 [5]^2 [4] [3]^2 [2]$ & $6$ \\
	
	6 & $2048$ & $16$ & $8$ & $5.00$ & $4.81$ & $0.19$ & $22$ & $[7]^{10}[6]^6 [5]^2 [4] [3]^2 [2]$ & $6$\\
		
	7 &  & $16$ & $10$ & & $4.83$ & $0.17$ & $23$ & $[9][8]^{13}[7]^3 [6]^2 [5] [4] [3] [2]$ & $5$ \\ 
	\hline
\end{tabular}
\end{table}

One can fix $M$, $N$, and $m$ while varying $\Lambda$ or $\sigma_n^2$. Table~\ref{table:vary} presents some results for $M=128$, $N=10$, $m=8$.

\begin{table}[h!]
\caption{Computed values for stochastic data with $M=128$, $N=10$, $m=8$, and $\lambda_j=\gamma^{j-1}$}
\label{table:vary}
\setlength{\tabcolsep}{0.15cm}
\centering
\begin{tabular}{ c c c c c c c l c}
\hline
No. & $\sigma_n^2$ & $\gamma$ & $\sum_{j=1}^M \lambda_j$ & Max $\Delta(N)$ & $\MSE(N)$ & $t$ & $\{s_k : k \in \bbra{t}\}$ & $\delta_{0.1}$ \\
\hline
1 & $0.05$ & $0.9$ & $10.00$ & $9.11$ & $0.89$ & $36$ & $[3]^{13} [2]^{18} [1]^5$ & $7$ \\ 

2 & 	   & $0.8$ & $5.00$ & $4.74$ & $0.26$ & $20$ & $[5]^{8}[4]^7 [3]^3 [2] [1]$ & $5$ \\
		
3 &  	   & $0.7$ & $3.33$ & $3.21$ & $0.12$ & $14$ & $[7]^{6}[6]^4 [5] [4] [3] [2]$ & $3$ \\
				
4 & $0.1$  & $0.9$ & $10.00$ & $8.68$ & $1.32$ & $32$ & $[3]^{20} [2]^8 [1]^4$ & $7$ \\ 
		
5 & 	   & $0.8$ & $5.00$ & $4.59$ & $0.41$ & $18$ & $[6]^{3} [5]^9 [4]^2 [3]^2 [2][1]$ & $4$ \\ 
		
6 &        & $0.7$ & $3.33$ & $3.14$ & $0.19$ & $13$ & $[8]^{5} [7]^3 [6][5][4][3][1]$ & $2$ \\
				
7 & $0.5$ & $0.9$ & $10.00$ & $7.01$ & $2.99$ & $22$ & $[5]^{7} [4]^7 [3]^3 [2]^3 [1]^2$ & $4$ \\
		
8 &      & $0.8$ & $5.00$ & $3.99$ & $1.01$ & $13$ & $[8]^{5} [7]^2 [6]^2 [5][4][3][2]$ & $2$ \\
		
\hline
\end{tabular}
\end{table}

Swapping $N$ and $m$ does not alter $\MSE(N)$, $t$ and $\{s_k : k \in \bbra{t}\}$. We keep $N$ and $m$ small compared to $M$ and use $m \leq N$ for smoother recovery, especially when few packets are available. Computation is longer for $m > N$ since, as $m$ increases, partitioning an $M$-dimensional space into subspaces of dimension $m$ requires exponentially more steps. The resulting plots confirm that the $\MSE$ reductions initially exhibit a larger fluctuation but converge relatively more rapidly when $m > N$. Figure~\ref{fig:interchange} illustrates the differences. 

\begin{figure}[t!]
\hspace{-1.2cm}
\includegraphics[width=1.2\textwidth]{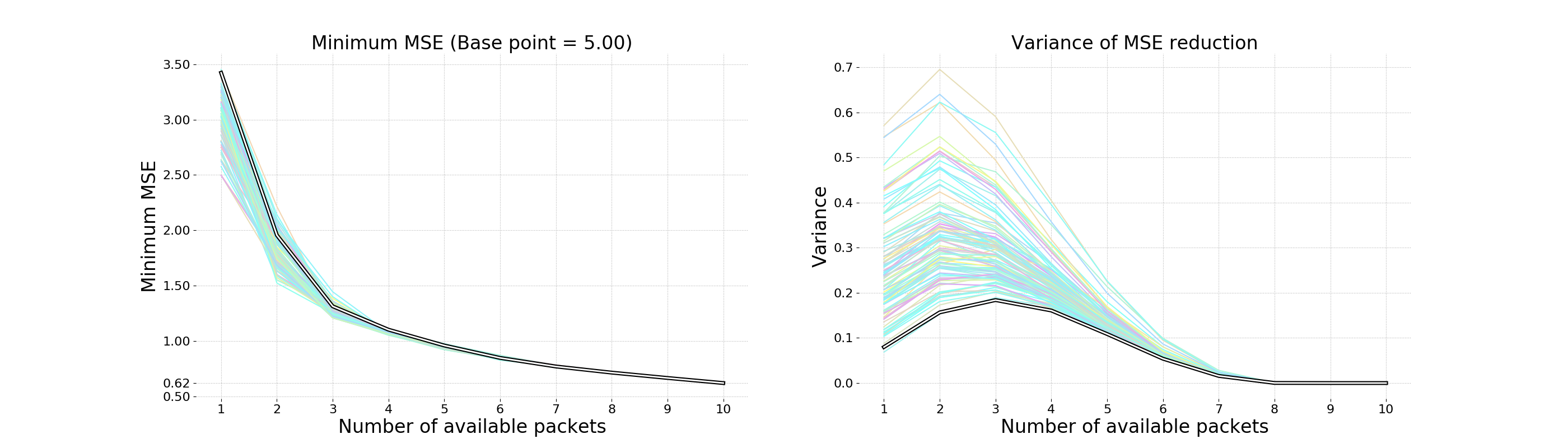}
\smallskip\par
\hspace{-1.2cm}
\includegraphics[width=1.2\textwidth]{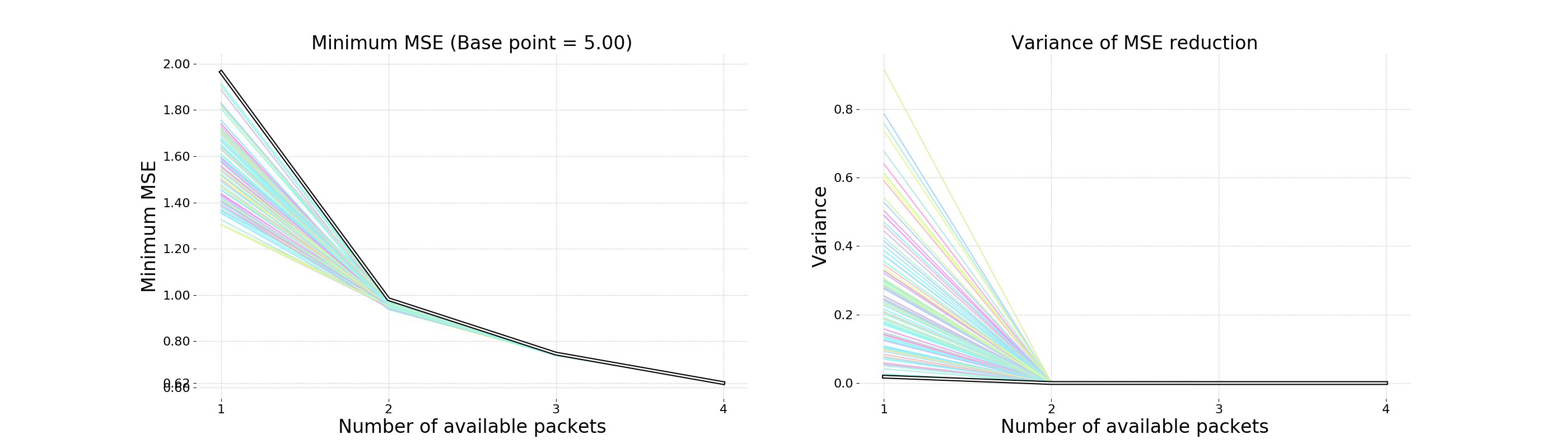}
\caption{Comparison when $N$ and $m$ are interchanged for a stochastic data with $M=1024$, $\sigma_n^2=0.1$, and $\lambda_j=0.8^{j-1}$. Above: $N=10$ and $m=4$. Below: $N=4$ and $m=10$.} \label{fig:interchange}
\end{figure}

\subsection{When $\lambda_j$ Decreases Linearly with $j$}

For $j \in \bbra{M}$, let $\displaystyle{\lambda_j = 1 - \frac{j-1}{M}}$ and $\Lambda=\diag(\lambda_1,\lambda_2,\ldots,\lambda_M)$. Such data is {\it linear}, with base point $(M+1)/2$, since the value of $\lambda_j$ decreases linearly with $j$. Compared to the data type in the preceding subsection, the output for linear data type shows higher variances in the $\MSE$ reduction among the subspace arrangements. The minimum $\MSE(\ell)$ values, however, are much closer to each other for any available $\ell$ packets. The coordinates are sampled more evenly as shown by the distribution of $s_k$s. Figure~\ref{fig:linear} presents the plots for the input $M=128$, $m=8$, $N=10$, and $\sigma_n^2=0.1$. Table~\ref{table:linear} has more examples.

\begin{figure}[h!]
\hspace{-1.2cm}
\includegraphics[width=1.2\textwidth]{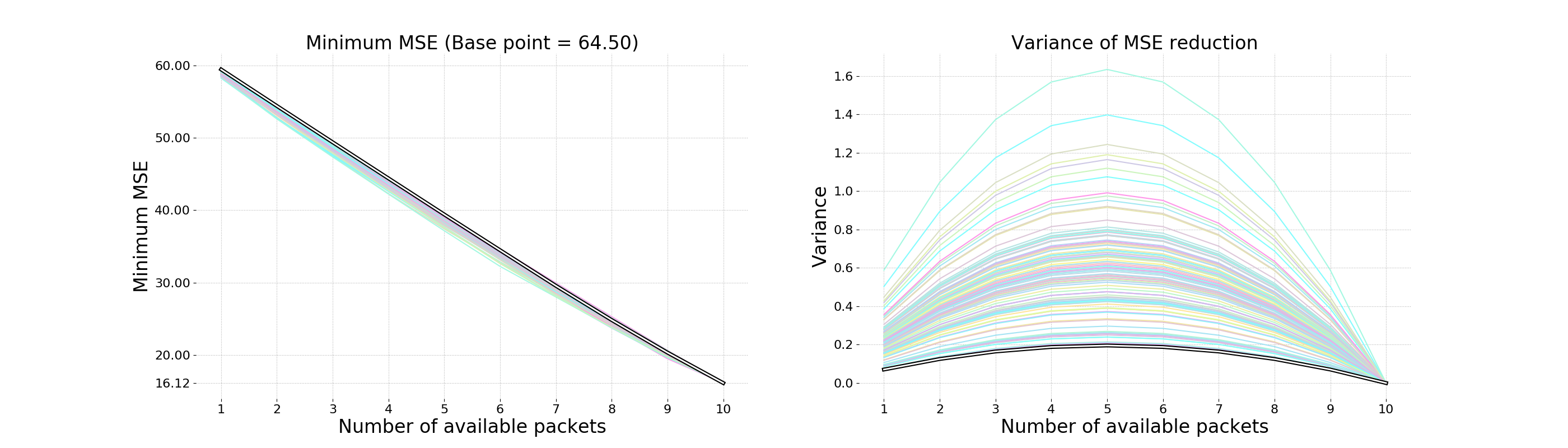}
\caption{Linear data with $M=128$, $m=8$, $N=10$, and $\sigma_n^2=0.1$}\label{fig:linear}
\end{figure}

\begin{table}[h!]
\caption{Computed values for linear data}
\label{table:linear}
\centering
\begin{tabular}{c c c c c c c c l c}
\hline
No. & $M$ & $N$ & $m$ & $\sigma_n^2$ & Max $\Delta(N)$ & $\MSE(N)$ & $t$ & $\{s_k : k \in \bbra{t}\}$ & $\delta_{0.1}$ \\
\hline
1 & $32$ & $4$ & $4$ & $0.05$ & $11.50$ & $5.00$ & $16$ & $[1]^{16}$ & $4$ \\
		
2 &  	& $8$ & $4$ & & $15.13$ & $1.37$ & $31$ & $[2][1]^{30}$ & $8$ \\ 
		
3 & $64$ & $8$ & $8$ & $0.05$ & $29.80$ & $2.70$ & $61$ & $[2]^3[1]^{58}$ & $8$ \\
		
4 &  &  &  & $0.1$ & $27.72$ & $4.78$ & $59$ & $[2]^5[1]^{54}$ & $8$ \\
		
5 & $128$  & $10$ & $8$ & $0.05$ & $51.60$ & $12.90$ & $80$ & $[1]^{80}$ & $10$ \\
		
6 & &  &  & $0.1$ & $48.38$ & $16.12$ & $80$ & $[1]^{80}$ & $10$ \\
		
7 & & &  & $0.5$ & $32.50$ & $32.00$ & $80$ & $[1]^{80}$ & $10$\\

\hline
\end{tabular}
\end{table}

\subsection{For Cyclostationary Data}

We also perform the computational analysis on the cyclostationary data with various $\gamma$ values. Recall that the nonzero diagonal entries $\lambda_j$ in $\Lambda$ is given by the formula in (\ref{eq:genlam}). The generated plots for a cyclostationary data with $M=128$, $m=8$, $N=10$, $\gamma=0.8$, and $\sigma_n^2=0.05$ form Figure~\ref{fig:cyclo}.

\begin{figure}[h!]
\hspace{-1.2cm}
\includegraphics[width=1.2\linewidth]{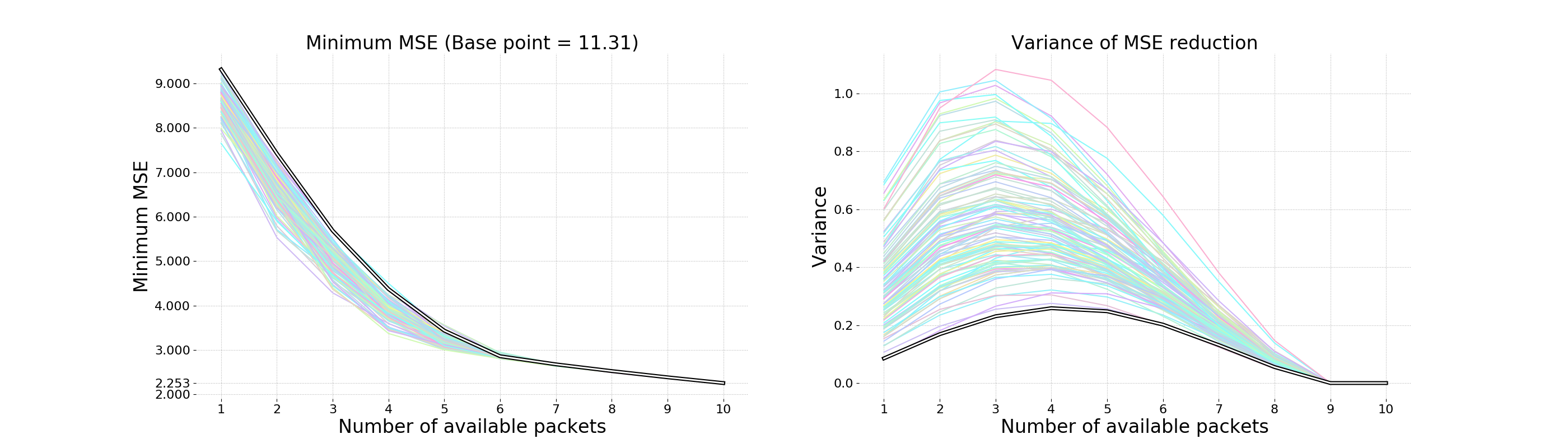}
\caption{Cyclostationary data with $M=128$, $m=8$, $N=10$, $\gamma=0.8$, and $\sigma_n^2=0.05$.}
\label{fig:cyclo}
\end{figure}

Table~\ref{table:cyclo} lists some computed values for the specified input parameters. The base point is $\sqrt{M}$ and the diagonal entries in $\Lambda$ are no longer monotonically nonincreasing. The presentation of $\{s_k\}$ for $k \in \bbra{M}$ must be adjusted accordingly since the threshold $t$ is meaningless here without a proper manipulation. Our strategy is to first order the diagonal entries in $\Lambda$ in a nonincreasing way and store the corresponding permutation $\tau$ of the indices. We then apply the method of determining $t$ and the $s_k$ for $k \in \bbra{t}$ as in the case of the typical stochastic data. Finally, we apply $\tau^{-1}$ to the set of indices to retrieve the correct index $k$ for each $s_k$. We use Entry 1 in Table~\ref{table:cyclo} to explain their presentation. The notation $[3]^3[2]^9[1]^6[0]^{29}[1]^6[2]^8[3]^3$ says that $s_k=3$ for $ k \in \bbra{1,3} \cup \bbra{62,64}$, $s_k=2$ for $k \in \bbra{4,12} \cup \bbra{54,61}$, $s_k=1$ for $k \in \bbra{13,18} \cup \bbra{48,53}$, and $s_k=0$ for $k \in \bbra{19,47}$.

\begin{table}[h!]
\caption{Computed values for cyclostationary data}
\label{table:cyclo}
\renewcommand{\arraystretch}{1.1}
\setlength{\tabcolsep}{0.1cm}
\centering
\begin{tabular}{c c c c c c c c l c}
\hline
No. & $M$ & $N$ & $m$ & $\gamma$ & $\sigma_n^2$ & Max $\Delta(N)$ & $\MSE(N)$ &  $\{s_k : k \in \bbra{M}\}$ & $\delta_{0.1}$ \\
\hline
1 & $64$ & $8$ & $8$ & $0.8$ & $0.05$ & $6.84$ & $1.16$ & $[3]^3[2]^9[1]^6[0]^{29}[1]^6[2]^8[3]^3$ & $6$ \\
		
2 & & $16$ & $4$ & & &  &  &  & $12$ \\ 
		
3 & $128$ & $10$ & $8$ & $0.9$ & $0.05$ & $9.90$ & $1.42$  &  $[3]^6[2]^9[1]^5[0]^{89}[1]^5[2]^8[3]^6$ & $8$ \\
		
4 &  &  &  & $0.7$ & & $8.33$ & $2.98$ &  $[2]^{12}[1]^{17}[0]^{71}[1]^{17}[2]^{11}$ & $8$ \\
		
5 & $256$  & $8$ & $8$ & $0.8$ &  $0.05$ & $11.02$ & $4.98$  & $[2][1]^{31}[0]^{193}[1]^{31}$ & $8$ \\
		
6 & &  &  & & $0.1$ & $9.83$ & $6.17$  & $[2]^9[1]^{15}[0]^{209}[1]^{15}[2]^8$ & $7$ \\
		
7 & & $16$ & $4$  & &  $0.5$ & $6.57$ & $9.43$  & $[3]^6[2]^6[1]^{3}[0]^{226}[1]^4[2]^6[3]^5$ & $5$\\
\hline
\end{tabular}
\end{table}

\subsection{An Adaptive Design}

A user may want to set a minimal acceptable number of available packets. Depending on the current channel situation, the user may prefer some flexibility in adapting the input parameters. Our implementation routine naturally reflects various requirements. To illustrate this point, consider a stochastic data with $M=1024$, $\lambda_j=0.8^{j-1}$ for $j \in \bbra{1024}$, $N=16$, $m=8$, and $\sigma_n^2=0.1$. The base point is $5.00$ and the best $\MSE(16)$ is $0.304$. Given a current channel, the user infers that, out of the $16$ possible packets, only up to $12$ arbitrary packets can be made available within a desirable time. With this additional constraint, the best $\MSE(12)$ is $0.366$. Imposing the smoothness condition, the best subspace arrangement for the original setup is generally no longer the best in the adapted situation. The user then adjusts accordingly by using this newly calculated best subspace arrangement. Figure~\ref{fig:truncate} allows for an easy comparison of the relevant plots.

%%%
\begin{figure}[h!]
	\hspace{-1.2cm}
	\includegraphics[width=1.2\textwidth]{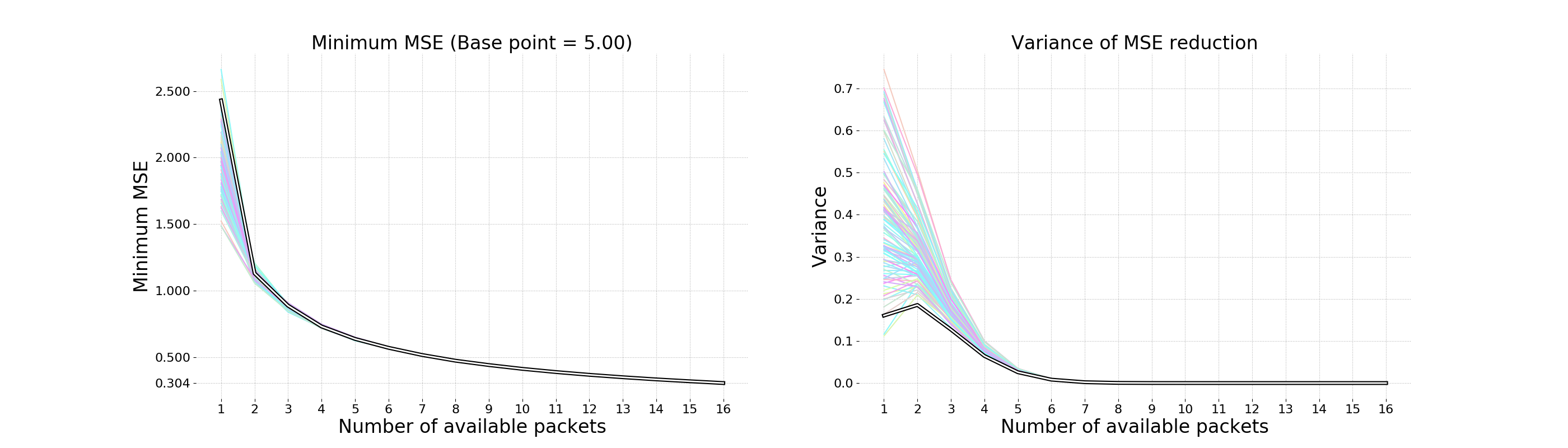}
	\smallskip\par
	\hspace{-1.2cm}
	\includegraphics[width=1.2\textwidth]{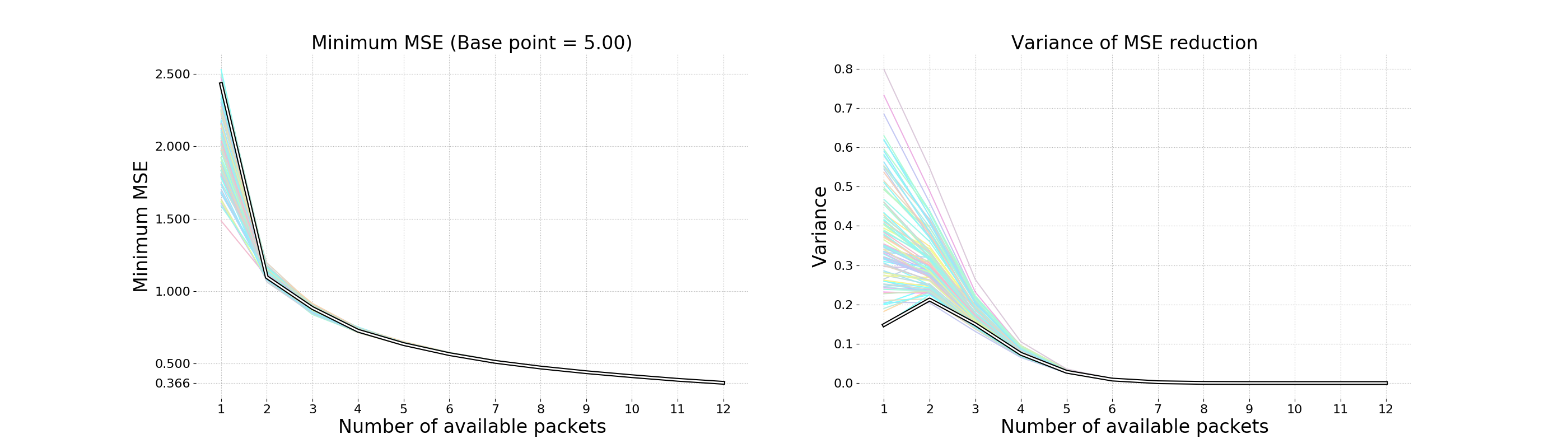}	
	\caption{An adaptive design for a stochastic data with $M=1024$, $\lambda_j=0.8^{j-1}$, $N=16$, $m=8$, and $\sigma_n^2=0.1$. Above: original setup. Below: only up to $12$ arbitrary packets are available.} \label{fig:truncate}
\end{figure}
%%%

\section{Connection to Grassmannian Packings}\label{sec:Kuty}

We now discuss how our approach relates to the work of Kutyniok \etal in~\cite{Kutyniok2009}. We begin with their setup. As in our estimation above, they use the linear minimum mean squared error estimation. The data is a random vector $\bx \in \bR^{M}$ of mean $0$ and covariance matrix $\BR_{xx}=\lambda \BI_{M}$ with $\lambda \triangleq \sigma_x^2$. The projections $\BP_{w_j}$ for $j \in \bbra{N}$ are general projections, not necessarily aligned with the standard basis. The estimation is based on the data's {\em fusion frame measurements} in the presence of additive white noise with possibilities of erasures. Their objective is to design a fusion frame which is robust against noise and erasures starting from erasures of any one subspace to those of any two or more subspaces.

Their analysis leads to three design criteria. First, in the presence of noise but without any erasure, the subspaces are best arranged in the form of a tight fusion frame with $A \cdot M = \sum_{k=j}^N m_j$ where $m_j$ is the dimension of subspace $\mW_j$. To robustly handle any one subspace erasure the subdimensions must be equal, \ie, $m_j=m$ for all $j$. Continuing to robustness against any two erasures, the Grassmannian packing of the subspaces yields the best error reduction in the estimation. With three or more subspaces unavailable, one should use subspace arrangement that forms an equidistance tight fusion frame with equal subdimension~\cite[Theorem~3.3]{Kutyniok2009}. Note that, in order to compare the results with our approach, the statement of the theorem needs to be refined. The $\MSE$ still depends on the number $r$ of erased subspaces and is not constant for all $r \geq 3$. The refined statement reads:

\begin{quote}
Let $\{\mW_k\}_{k=1}^N$ be an equidistance tight fusion frame with $\dim(\mW_k)=m$ for all $k \in \bbra{N}$. Then the $\MSE(\lambda=\sigma_x^2,\sigma_n^2,N-r)$ due to $r$ subspace erasures for each $r \in \bbra{3,N-1}$ depends only on $r$. 
\end{quote}

The work of Kutyniok \etal in~\cite{Kutyniok2009} made use of two simplifying assumptions that, given the results that we have obtained above, can be removed to yield better $\MSE$ performance. First, their choice of using the matrix $\BI - \BE$, accounting for the loss of data, to avoid recalculating $\BR_{ee}$ on every occasion (\cite[Section 3]{Kutyniok2009}) degrades the performance of the estimation process. A more careful analysis on the matrix $\BR_{ee}$ given information about the specifics of any set of $\ell$ available packets allows for a sharp determination of the achievable $\MSE$ for each particular instance. Second, considering only the case of $\BR_{xx}=\sigma_x^2 \BI_M$ does not reflect many realistic situations. It is more common to have data with $\BR_{xx}=\Psi \Lambda \Psi^{\top}$ where $\Lambda=\diag(\lambda_1,\lambda_2,\ldots,\lambda_M)$ depends on the exact, and usually given or estimated, statistical characteristics of the data.

We now follow the setup in~\cite{Kutyniok2009} with $\sigma_x^2 \BI$ replaced by $\Lambda$ and retrace the analysis, starting from the no erasure case onward. We form the composite measurement matrix $\bz_{\rm combi}$ by concatenating the $\bz_k$ for $k \in \bbra{N}$ and define the sum of the projections by using the composite basis matrix $\BU :=(\BU_1 | \ldots | \BU_N)$. When there is no erasure, the error covariance matrix is
$\BR_{ee}=\displaystyle{\left(\BR_{xx}^{-1} + \frac{1}{\sigma_n^2} \BU \BU^{\top}\right)^{-1}}$. Let $\phi_j^{-1}$ for $j \in \bbra{M}$ be the $j$-th eigenvalue of $\BR_{ee}^{-1}$. Hence, $\Tr(\BR_{ee})=\sum_{j=1}^M \phi_j$. 
For each $j$, we have $\displaystyle{
	\frac{1}{\lambda_j} + \frac{A}{\sigma_n^2} \leq \frac{1}{\phi_j} \leq \frac{1}{\lambda_j} + \frac{B}{\sigma_n^2}}$ by (\ref{eq:frame}).
This implies $\displaystyle{
\sum_{j=1}^{M} \frac{\lambda_j \sigma_n^2}{\sigma_n^2 + \lambda_j B} 
\leq \sum_{j=1}^{M} \phi_j
\leq \sum_{j=1}^{M} \frac{\lambda_j \sigma_n^2}{\sigma_n^2 + \lambda_j A}}$. Highlighting the {\it noise-to-signal ratio}, 
$\displaystyle{
\sum_{j=1}^M \frac{\lambda_j}{1+\left(\frac{\lambda_j}{\sigma_n^2} \right)B} \leq \Tr(\BR_{ee})=\sum_{j=1}^M \phi_j \leq 
\sum_{j=1}^M \frac{\lambda_j}{1+\left(\frac{\lambda_j}{\sigma_n^2} \right)A}}$. Let $L \triangleq \sum_{k=1}^N m_k$. Then the minimal $\MSE$ is achieved when $A=B$, \ie, when the fusion frame is tight. Thus, $A=\displaystyle{\frac{L}{M}}$ and 
\begin{equation}\label{eq:MSE_N_A}
\MSE(\Lambda,\sigma_n^2,N,\mbox{tight fusion frame})=\sum_{j=1}^M \frac{\lambda_j}{1+\left(\frac{\lambda_j}{\sigma_n^2} \cdot \frac{L}{M}\right)}.
\end{equation}

\begin{remark}\label{remark:simplyframe}
If we simply have a frame without requiring $\sum_{j=1}^N \BP_k \geq A \BI_{M}$, then $\BR_{ee}^{-1}$ is diagonal with entries  $\phi_j^{-1}=\displaystyle{\frac{1}{\lambda_j}+ \frac{s}{\sigma_n^2}}$, making 
$\displaystyle{\phi_j= \frac{\lambda_j \sigma_n^2}{\sigma_n^2 + s\lambda_j}}$ for $j \in \cP_s$. 
Hence, $\displaystyle{
\MSE(\Lambda,\sigma_n^2,N)=\sum_{k=1}^{M} \phi_k =\sum_{s=0}^{N} 
\sum_{j \in \cP_s} \frac{\lambda_j \sigma_n^2}{\sigma_n^2 + s\lambda_j}}$, which was established earlier in (\ref{eq:N_align}). 
\end{remark}

It is clear that removing the requirement of using a fusion frame yields lower $\MSE$ estimators when $\BR_{xx}=\Lambda$. When $\BR_{xx}=\sigma_x^2 \BI$, however, we have already seen in the derivation of (\ref{eq:MSE_tight}) that the minimum value for $\MSE(\sigma_x^2,\sigma_n^2,N)$ is indeed achieved when an $A$-tight fusion frame is used.

Continuing our analysis on the general erasure model described in~\cite[Section 3]{Kutyniok2009}, still with $\BR_{ee}=\Lambda$, let $r$ subspaces out of the $N$ forming the frame be erased. Assume that all packets have equal dimension $m$. Let $\Xi \subset \bbra{N}$ be the set of indices of the {\em erased} subspaces and let $\BE$ be the corresponding (symmetric) $L \times L$ block diagonal erasure matrix. Its $j$-th diagonal block is $\BI_{m}$ if $j \in \Xi$ and is an $m \times m$ zero matrix if $j \notin \Xi$. Then the composite measurement vector with erasures is $\tilde{\bz}= (\BI - \BE) \bz$. Thus, in $\tilde{\bz}$ the measurement vectors associated with the erased subspaces are set to $\0$. 

The estimate of $\bx$ is $\tilde{\bx}=\BF \tilde{\bz}$. The error covariance matrix $\widetilde{\BR}_{ee}$ for this estimate is
\[
\EE[(\bx - \tilde{\bx})(\bx - \tilde{\bx})^{\top}]=
\EE[(\bx - \BF(\BI - \BE) \bz)(\bx - \BF(\BI - \BE) \bz)^{\top}].
\]
Let $\BR_{ee} \triangleq \BR_{xx} - \BR_{xz} \BR_{zz}^{-1} \BR_{zx}$ and $\overline{\BR}_{ee} \triangleq \BR_{xz} \BR_{zz}^{-1} \BE \BR_{zz} \BE^{\top} \BR_{zz}^{-1} \BR_{zx}$. We use them to define   
$\widetilde{\BR}_{ee} \triangleq \BR_{ee}+\overline{\BR}_{ee}$. Now we minimize the trace of $\overline{\BR}_{ee}$, which can be written as
\begin{equation}\label{eq:Rtil}
\BR_{xx} \BU (\sigma_n^2 \BI_{L} + \BU^{\top} \BR_{xx} \BU)^{-1} \BE (\sigma_n^2 \BI_{L} + \BU^{\top} \BR_{xx} \BU) 
\BE^{\top} (\sigma_n^2 \BI_{L} + \BU^{\top} \BR_{xx} \BU)^{-1} \BU^{\top} \BR_{xx}.
\end{equation}
Let $\displaystyle{ \alpha_j \triangleq \frac{\lambda_j}{\sigma_n^2+ \frac{L}{M} \lambda_j}}$. By Proposition~\ref{propW} with $\BA=\sigma_n^2 \BI_{L}$, $\BC=\BU^{\top}$, and $\BD=\BR_{xx} \BU$, we write the symmetric matrix  $\displaystyle{\BR_{zz}^{-1}=(\sigma_n^2 \BI_{L} + \BU^{\top} \BR_{xx} \BU)^{-1}}$ as
\begin{equation*}\label{eq:Inv}
\frac{1}{\sigma_n^2}\BI_{L} - \frac{1}{\sigma_n^{4}} \BU^{\top} 
\left(\BI_{M} + \frac{L}{M \sigma_n^2} \BR_{xx} \right)^{-1} \BR_{xx} \BU =
\frac{1}{\sigma_n^2}\BI_{L} - \frac{1}{\sigma_n^{2}} \BU^{\top} 
\diag(\alpha_1, \ldots,\alpha_M) \BU.
\end{equation*}
Let 
$\displaystyle{
\BY \triangleq \BR_{xx} \BU \left(\frac{1}{\sigma_n^2} \BI_{L} - 
\frac{1}{\sigma_n^2} \BU^{\top} ~\diag(\alpha_1, \ldots,\alpha_M)~ \BU\right)}$. 
Performing the calculation, 
\begin{align*}
\BY &=\frac{1}{\sigma_n^2} 
~\diag(\lambda_1,\ldots,\lambda_{M})~\BU - 
\frac{L}{M \sigma_n^2} 
~\diag(\lambda_1,\ldots,\lambda_{M})~
\diag(\alpha_1, \ldots,\alpha_M)~\BU \\
&=\diag(\alpha_1, \ldots,\alpha_M)~\BU.
\end{align*}
Hence, 
$\displaystyle{
\overline{\BR}_{ee}= \BY \BE (\sigma_n^2 \BI_{L} + \BU^{\top} \BR_{xx} \BU) \BE^{\top} \BY^{\top}}$. Since $\BU \BE \BU^{\top}=
\displaystyle{\sum_{\ell \in \Xi} \BP_{\ell}}$, $\overline{\BR}_{ee}$ is given by 
\begin{multline*}
\diag(\alpha_1, \ldots,\alpha_M)~ 
\sigma_n^2 \left(\sum_{\ell \in \Xi} \BP_{\ell}\right) ~
\diag(\alpha_1, \ldots,\alpha_M)\\
+
\diag(\alpha_1, \ldots,\alpha_M) 
\left(\sum_{\ell \in \Xi} \BP_{\ell}\right) ~
\diag(\lambda_1, \ldots, \lambda_{M})~
\left(\sum_{\ell \in \Xi} \BP_{\ell}\right)~
\diag(\alpha_1, \ldots,\alpha_M).
\end{multline*}
Rearranging for better computation of the trace, we conclude that $\overline{\BR}_{ee}$ is given by 
\begin{equation}\label{eq:simpler}
\diag(\alpha_1, \ldots,\alpha_M)
\left(\sigma_n^2 \sum_{\ell \in \Xi} \BP_{\ell} + 
\sum_{\ell \in \Xi} \BP_{\ell}
~\diag(\lambda_1, \ldots, \lambda_{M})~
\sum_{j \in \Xi} \BP_j\right)
\diag(\alpha_1, \ldots,\alpha_M).
\end{equation}

Let $\MSE_0 \triangleq \MSE(\Lambda,\sigma_n^2,N,\mbox{tight fusion frame})$ in (\ref{eq:MSE_N_A}) and $\overline{\MSE}:=\Tr(\overline{\BR}_{ee})$. Our $\MSE(\Lambda,\sigma_n^2, N-r,\mbox{tight fusion frame})$ is therefore given by $\Tr[\widetilde{\BR}_{ee}]=\MSE_0+ \overline{\MSE}$.

\begin{remark}
It is immediate to verify that our results include the corresponding results in~\cite{Kutyniok2009} as special cases. One simply replaces each $\lambda_j$ in (\ref{eq:MSE_N_A}) and (\ref{eq:simpler}) above with $\sigma_x^2$. 
In the former, one arrives at
	\[
	\sum_{k=1}^{M} \frac{\sigma_x^2}{1+\frac{A \sigma_x^2 }{\sigma_n^2}}=
	\frac{M \sigma_x^2 \sigma_n^2}{\sigma_n^2+A \sigma_x^2} \mbox{ with } A=\frac{L}{M},
	\]
	which is the $\MSE_0$ in~\cite[Equation (7)]{Kutyniok2009}. In the latter, the result is \cite[Equation (9)]{Kutyniok2009}
	\[
	\alpha^2 \Tr\left(\sigma_n^2 \sum_{j \in \Xi} \BP_j + \sigma_x^2 \left( \sum_{j \in \Xi} \BP_j\right)^2\right) \mbox{ with } \alpha \triangleq\frac{\sigma_x^2}{\sigma_n^2+A \sigma_x^2}.
	\]
\end{remark}

Hence, following the approach of Kutyniok \etal using a more general $\BR_{xx}$ confirms that their results are obtained when $\BR_{xx}=\sigma_x^2 \BI_{M}$. It is however clear that in all applicable instances, the $\MSE(\Lambda,\sigma_n^2,\ell=N-r)$ in (\ref{eq:l_align}) is more precise since, given any $r$ out of $N$ packets missing, the error covariance matrix $\BR_{ee}$ is recalculated. It is here where Grassmannian packings no longer provide any benefit. The better $\MSE$ performance indeed requires more yet still reasonable computations. 

In designing packets of measurement for estimating the unknown vector $\bx$ to be robust against noise and erasures, starting from one erasure onward, we have the following conclusions. First, the best setup for recovery when all packets are available is as follows. When $\BR_{xx}=\sigma_x^2 \BI$, we can indeed use a tight fusion frame. When $\BR_{xx}=\Lambda$, one first determines the best values of $s_1,s_2,\ldots, s_M$ that maximize the $\Theta$ or $\Upsilon$ function. Here we do not even have a fusion frame if there is a $t<M$ such that $s_j =0$ for all $j \in \bbra{t+1,M}$. 
Given that there is one subspace erasure, we have the following strategy: If $\BR_{xx}=\sigma_x^2 \BI$, then we ensure that all subspaces have equal dimension $m$. When $\BR_{xx}=\Lambda$, compute for $\MSE(\Lambda,\sigma_n^2,N-1)$ using the information on the particular missing packet. Our work still assumes that all subspaces are of equal dimension $m$. It remains an interesting possibility to consider letting the subspace $\mW_j$ be of dimension $m_j$ which is determined to be a function of $\lambda_j$ for $j \in \bbra{M}$. In the framework of \cite{Kutyniok2009}, when any two subspaces are erased and $\BR_{xx}=\sigma_x^2 \BI$, one uses suitable Grassmannian packings. If, subsequently, $r \geq 3$ erasures occur, the formula for the corresponding $\MSE$ depends only on $r$. On the other hand, when $\BR_{xx}=\Lambda$ and $r \geq 2$ erasures take place, one should use the formula for  $\MSE(\Lambda,\sigma_n^2,N-r)$, incorporating the exact set $\Xi$ of the missing packets in the computation, to evaluate the system's $\MSE$ performance.

\section{Conclusion and Other Directions}\label{sec:conclusion}

We put forward a general sensing method that produces holographic representations of data. Packets that encode the information are designed to be equally important and the progressive recovery of the unknown vector from those packets will have as smooth decreasing error profiles as possible. Thus, the quality of recovery depends on the number of available packets, regardless of the order in which they arrive. An optimality analysis based on the least-squares estimation theory is supplied in detail. We are currently investigating if other known techniques, such as network coding and projection matrix designs for compressive sensing, can further improve our method.

To gain significantly from our holographic sensing, the data must be useful at various quality levels. Our approach is well suited for storage and distributed retrieval of information such as speech, audio, image, video, and volumetric data where there is unpredictable delay in gathering the full data and, hence, an early degraded preview can be very useful to decide whether or not to proceed with the retrieval process.

\section*{Acknowledgements}

Singapore National Research Foundation and Israel Science Foundation joint program NRF2015-NRF-ISF001-2597 and Nanyang Technological University Grant Number M4080456 support the work of the authors.

\end{document}